\documentclass[a4paper,11pt]{article}
\usepackage[
  top=1in,
  bottom=1in,
  left=1in,
  right=1in
]{geometry}
\usepackage[utf8]{inputenc}
  
\usepackage{graphicx}
\usepackage{pstricks}
\usepackage{appendix}
\usepackage{dsfont}
\usepackage{amsmath}
\usepackage{amssymb}
\usepackage{amsthm}
\usepackage{mathtools}
\usepackage{comment}
\usepackage[shortlabels]{enumitem}
\usepackage{tikz,lipsum,lmodern}
\usepackage[most]{tcolorbox}
\usepackage{bbm}
\usepackage{algorithm}
\usepackage{algpseudocode} 
\usepackage{float}  
\usepackage{enumitem}
\usepackage{amsmath}
\usepackage{listings}
\usepackage{verbatim}
\usepackage{soul}
\usepackage{subcaption}

 \usepackage{natbib}   
\newcommand{\explain}{\nn&\quad\blacktriangleright}
\newcommand{\nn}{\nonumber}

\usepackage{xcolor}
\usepackage{blindtext}

\usepackage[nottoc]{tocbibind}

\usepackage[dvipsnames]{xcolor}
\usepackage[
    linktoc=page,
    colorlinks=true,
    linkcolor=RoyalBlue,
    citecolor=BrickRed,
    urlcolor=ForestGreen
]{hyperref}
\usepackage{amsthm}

\newtheorem{theorem}{Theorem}[section]

\newtheorem{corollary}{Corollary}[section]
\newtheorem{lem}{Lemma}[section]

\newtheorem{rem}{Remark}[section]
\newtheorem{assumption}{Assumption}[section]

 \numberwithin{equation}{section}

\DeclareMathOperator{\diag}{diag}

\newcommand{\sign}{\mathrm{sign}}
\allowdisplaybreaks

\title{One-Bit Clustering for Two Component Sub-Gaussian Mixture Models}

\begin{document}
\author{Junren Chen\thanks{University of Maryland, College Park.} \and Yun Yang\thanks{University of Maryland, College Park.}}

\maketitle

\begin{abstract}
  Clustering is a fundamental problem in statistics and machine learning. We propose the first one-bit clustering method for two-component sub-Gaussian mixture models. The method uses only one bit per entry of each sample obtained via a dithered quantizer. Under a mild non-spikiness condition on the cluster centers, we show that a variant of Lloyd’s algorithm achieves a misclassification rate that decays exponentially with a signal-to-noise ratio comparable to that in the unquantized setting. This result further implies exact recovery under an explicit separation condition, which exceeds the optimal threshold for unquantized data by only a logarithmic factor. When the dimension $p$ is sufficiently large, the non-spikiness condition can be enforced by applying a random rotation using a Haar distributed matrix prior to quantization. In particular, it holds with high probability when $p \gtrsim 1$ for partial recovery and $p \gtrsim \log n \log\log n$ for exact recovery, where $n$ is the sample size. We also establish a minimax lower bound, showing that the misclassification rate and separation condition exhibit sharp constants in general. Numerical results are provided to corroborate the theory and demonstrate the efficacy of the proposed method.  
\end{abstract}

\section{Introduction}

 Clustering is arguably a fundamental problem in statistics, signal processing, and machine learning. 
Suppose that the samples $\{X_i\}_{i=1}^n \in \mathbb{R}^p$ come from a few (unknown) centers $\{\theta_i\}_{i=1}^K\subset \mathbb{R}^p$, the goal of clustering is to learn the labels and group the samples from the same center together. The classical setting requires access to the samples $X_i \in \mathbb{R}^p$, which may be unrealistic in some modern applications such as distributed learning where the communication cost can be prohibitive and low-bit data may be used instead. This consideration raises a fundamental question:
\[
\textit{Is accurate or even exact clustering possible under one-bit quantization?}
\]
Specifically, in this paper we ask whether accurate clustering is possible using one bit per entry. While similar questions have been studied in compressed sensing \cite{jacques2013robust,plan2012robust,dirksen2021non}, matrix completion \cite{davenport20141,chen2022high}, mean estimation \cite{kipnis2022mean,cai2022distributed}, covariance estimation \cite{dirksen2022covariance,chen2025parameter}, and phase retrieval \cite{domel2022phase,chen2024one}, among others, to the best of our knowledge, this question remains unexplored in the context of clustering.


Our paper provides an affirmative answer to this question in the setting of a two-component symmetric sub-Gaussian mixture model centered at $\pm \theta$ for some unknown $\theta \in \mathbb{R}^p$:
\begin{align}\label{model2mixture}
    X_i = \eta_i \theta + \varepsilon_i,\quad i = 1,2,\cdots, n,
\end{align}
where $\eta = (\eta_1,\cdots,\eta_n)^T \in \{-1,1\}^n$ are the labels to be estimated, and $\varepsilon_i$ are sub-Gaussian noise vectors (see Assumption \ref{subgaussian}). We develop an estimator $\hat{\eta}$ that uses only one bit per entry from $\{X_i\}_{i=1}^n$ and achieves low, or even zero, misclassification rate
\begin{align*}
    \ell(\hat{\eta},\eta) = \frac{1}{n}\min\bigg\{\sum_{i=1}^n \mathbf{1}(\hat{\eta}_i \ne \eta_i),\;\sum_{i=1}^n \mathbf{1}(\hat{\eta}_i \ne -\eta_i)\bigg\},
\end{align*}
referred to as \emph{partial recovery} and \emph{exact recovery} \cite{giraud2019partial,fei2018hidden,chen2021cutoff,ndaoud2022sharp}, respectively.

Our quantization scheme uses uniform dithering, that is, we add uniform ``noise'' distributed as ${\rm Unif}[-\lambda,\lambda]$ to the samples before applying the one-bit sign quantizer. We refer to $\lambda$ as the dithering level. The benefits of dithering were observed early \cite{roberts1962picture} and have recently received significant attention in the statistical estimation literature \cite{dirksen2022covariance,chen2022high,xu2020quantized}. We will further elucidate the intuition and trade-offs underlying this quantization scheme. We also argue that direct quantization without dithering does not retain sufficient information for clustering in sub-Gaussian mixtures (cf. Remark \ref{rem:dither}).


The work most relevant to our paper is \cite{ndaoud2022sharp}, which establishes the minimax-optimal misclassification rate for two symmetric Gaussian mixtures, that is, model~\eqref{model2mixture} with $\varepsilon_i \sim N(0,\sigma^2 I_p)$:
\begin{align}\label{sharppari}
      \exp\Big(-\frac{1+o(1)}{2}r_{n,\sigma}^2\Big)\le \inf_{\hat{\eta}}\sup_{\eta \in\{-1,1\}^n}\mathbb{E}[\ell(\hat{\eta},\eta)] \le \exp\Big(-\frac{1-o(1)}{2}r_{n,\sigma}^2\Big)
\end{align}
under $r_{n,\sigma}\gtrsim 1$, where $r_{n,\sigma}:=\frac{\|\theta\|_2^2/\sigma^2}{\sqrt{\|\theta\|_2^2/\sigma^2+p/n}}$ is the signal-to-noise ratio (SNR). Note that the sharp constant in the exponent is characterized, and the upper bound is achieved by a variant of Lloyd's algorithm \cite{peng2007approximating,lu2016statistical}. 
The bound~\eqref{sharppari} further implies exact recovery under the sharp separation condition
\begin{align}
    \|\theta\|_2^2 \ge (1+\epsilon)\sigma^2\bigg(1+\sqrt{1+\frac{2p}{n\log n}}\bigg)\log n.\label{existsepara}
\end{align}
Some of these sharp results have been extended to mixtures with more than two components using semidefinite programming \cite{chen2021cutoff} and spectral clustering \cite{loffler2021optimality,abbe2022lp}.

We establish similar results under dithered one-bit quantization. In a nutshell, our estimator $\hat{\eta}_{1b}$, computed via a variant of Lloyd's algorithm with a small number of iterations, uses one bit per entry of $\{X_i\}_{i=1}^n$ and attains the misclassification rate
$
\ell(\hat{\eta}_{1b},\eta)\le \exp\Big(-\frac{(1-o(1))r_{n,\lambda}^2}{2}\Big)
$
under $r_{n,\lambda}\gtrsim 1$, where $r_{n,\lambda}:=\frac{\|\theta\|_2^2/\lambda^2}{\sqrt{\|\theta\|_2^2/\lambda^2+p/n}}$ is our new signal-to-noise ratio with $\sigma$ in $r_{n,\sigma}$ replaced by the dithering level $\lambda$. Under $\varepsilon_i\sim N(0,\sigma^2)$, $\lambda$ can be chosen on the order of $\sigma\sqrt{\log (np)}$, only a logarithmic factor larger than $\sigma$. Hence, our rate is only slightly worse than \eqref{sharppari}, indicating that the quantization incurs little information loss.
The misclassification bound immediately implies exact recovery under an appropriate separation condition, in contrast to most existing one-bit estimation results that only guarantee approximate recovery (cf. Remark \ref{rem:exact}). We also establish a lower bound showing that our misclassification rate exhibits a sharp constant in general. While the developments appear parallel to \cite{ndaoud2022sharp}, substantially different and new techniques are required to establish these results. Moreover, there is an additional subtle interaction among the quantization, dimension, and the spikiness of the center (cf. Remarks \ref{rem:spiky1}, \ref{rem:spiky2}), which appears to be a novel phenomenon in the area of one-bit learning and motivates a Haar matrix preprocessing step.


\noindent
\textbf{Notation.} Let $[m]=\{1,2,\cdots,m\}$ for an integer $m$. 
Let $\sign(a)=1$ for $a\ge 0$ and $\sign(a)=-1$ otherwise; this is applied elementwise to vectors. Let $\|v\|_\infty=\max_{i}|v_i|$ and $\|v\|_q=(\sum_i |v_i|^q)^{1/q}$ denote the max norm and $\ell_q$ norm of a vector $v\in \mathbb{R}^n$, respectively, and let $\|M\|_{op}$ denote the operator norm of a matrix $M$. 
We use $C,c,C_0,C_1,C',\cdots$ to denote universal constants whose values may vary from line to line. We write $T_1=O(T_2)$ (or $T_1\lesssim T_2$) to denote $T_1\le C T_2$, and write $T_1=\Omega(T_2)$ (or $T_1\gtrsim T_2$) to denote $T_1\ge C T_2$. We also write $T_1=\Theta(T_2)$ (or $T_1\asymp T_2$) if both $T_1=O(T_2)$ and $T_1=\Omega(T_2)$ hold.
The sub-Gaussian norm of a random variable $X$ is defined as $\|X\|_{\psi_2}=\inf\{K>0:\mathbb{E}[\exp(X^2/K^2)]\le 2\}$. More notation will be introduced as needed.

\noindent
\textbf{Overview.} In Section \ref{sec:main}, we first introduce our quantization procedure and clustering algorithm, and then present our one-bit clustering recovery guarantees, along with an overview of the technical proofs. In Section \ref{sec:lower}, we establish a minimax lower bound to demonstrate the sharpness of our results. Section \ref{sec:experi} provides numerical examples to corroborate our theory, and Section \ref{sec:conclu} concludes the paper with several remarks. The complete proofs and most technical lemmas are deferred to the appendix.

\section{Main Results}\label{sec:main}
Throughout the paper, we adopt the sub-Gaussian mixture model~\eqref{model2mixture} and impose the following assumption.
\begin{assumption}\label{subgaussian}
    $\{\varepsilon_i\}_{i=1}^n$ are i.i.d.\ noise vectors with independent, symmetric, zero-mean entries satisfying the $\sigma$-sub-Gaussian condition
$\mathbb{E}[\exp(t\varepsilon_{ij})] \le \exp\!\big(\tfrac{t^2\sigma^2}{2}\big)$ for all $t\ge 0$ and some $\sigma>0$.
\end{assumption}

We collect the unquantized samples in $Y:=[X_1,\cdots,X_n]$. Independent of $\{\varepsilon_i\}_{i=1}^n$, we draw dithers $\{\tau_i\}_{i=1}^n$ with entries i.i.d.\ uniformly distributed over $[-\lambda,\lambda]$, that is, $\tau_i\sim{\rm Unif}[-\lambda,\lambda]^p$, and quantize $X_i$ to $\dot{X}_i = \sign(X_i+\tau_i)\in \{-1,1\}^{p}$. This dithered quantization scheme has been used to address mean and covariance estimation under coarse quantization \cite{abdalla2026robust,dirksen2022covariance,chen2025parameter}. The intuition is that the expectation
\begin{align}
    \mathbb{E}_{\tau\sim {\rm Unif}[-\lambda,\lambda]}[\lambda \sign(a+\tau)]
    =T_{\lambda}(a)
    :=
    \begin{cases}
    a\,,\quad&\textrm{if }|a|\le \lambda,\\
    \lambda\sign(a)\,,\quad& \textrm{if }|a|>\lambda,
    \end{cases}
    \label{bvtrade}
\end{align}
is a truncation of the input $a$ (cf.\ Lemma~\ref{lem:expect}). Therefore, for sufficiently large $\lambda$, $\lambda \dot{X}_i$ can serve as a good surrogate for $X_i$ and retain sufficient information for clustering. We collect all the one-bit samples in a matrix $\dot{Y}=[\dot{X}_1,\cdots,\dot{X}_n]$.


To motivate our algorithm, we note that the minimax optimal procedure in \cite{ndaoud2022sharp} consists of two steps: (i) compute the top eigenvector $\hat{v}$ of $H(Y^TY)$ and set $\eta^0=\sign(\hat{v})$, where the hollowing operator $H:\mathbb{R}^{n\times n}\to\mathbb{R}^{n\times n}$ removes the diagonal of a square matrix, that is, $H(M)=M-\diag(M)$; (ii) generate a sequence of estimates via a projected power iteration $\eta^{k+1} = \sign(H(Y^TY)\eta^k)$ for $k\ge 0$.\footnote{The $\sign$ here is a retraction that maps an estimate back to the parameter space $\{-1,1\}^n$.} Our observation is that the hollowed Gram matrix $H(Y^TY)$ serves as a sufficient statistic for both steps. Combining this with the intuition that $\lambda \dot{X}_i$ acts as a surrogate for $X_i$, we propose to use $H(\lambda^2 \dot{Y}^T \dot{Y})$ to replace the unavailable $H(Y^TY)$. Since $\lambda^2$ is absorbed into the retraction $\sign(\cdot)$, we arrive at Algorithm~\ref{alg:2lloyd}.

\begin{rem}\label{rem:bias-variance}
   While Algorithm~\ref{alg:2lloyd} does not require the knowledge of $\lambda$, it is a key parameter that affects performance. To reduce bias, $\lambda$ should be large enough to dominate most entries of $\{X_i\}_{i=1}^n$ (cf.\ Equation \eqref{bvtrade}). However, a larger $\lambda$ also induces higher variance in $\lambda \dot{Y}$ and thus slower concentration. From another perspective,  as $\lambda \to \infty$, $\dot{Y}$ converges  to the non-informative ${\rm Unif}\{-1,1\}$ in distribution. Therefore, choosing $\lambda$ involves a bias--variance trade-off. 
\end{rem}
 

\begin{algorithm}[t]   
\caption{One-bit Lloyd's Algorithm for 2-Component Symmetric Sub-Gaussian Mixtures}\label{alg:2lloyd}
\begin{algorithmic}[1]   
  \Require Quantized samples $\dot{Y}=[\dot{X}_1,\cdots,\dot{X}_n]$, iteration number $T_0$  

  \State \textbf{Step 1: Initialization}
  \State Let $\hat{v}$ be the leading eigenvector of $H(\dot{Y}^T\dot{Y})$ and let $\hat{\eta}^0 = \sign(\hat{v})$

  \State \textbf{Step 2: Refinement}
  \For{$k = 0,1,2,\dots,T_0-1$} 
   \begin{align}
          \hat{\eta}^{k+1}=\sign(H(\dot{Y}^T\dot{Y})\hat{\eta}^k) \label{lloydequa}
    \end{align}
  \EndFor

  \Ensure $\hat{\eta}^{T_0}$
\end{algorithmic}
\end{algorithm}

\subsection{Recovery Guarantees}

 Theorem~\ref{maintheorem} below shows that Algorithm~\ref{alg:2lloyd} attains a misclassification rate that decays exponentially with the signal-to-noise ratio $r_{n,\lambda}:= \frac{\|\theta\|_2^2/\lambda^2}{\sqrt{\|\theta\|_2^2/\lambda^2+p/n}}$. An interesting feature is the appearance of the \emph{spikiness} of the center, defined by $\mu(\theta)=\frac{\|\theta\|_{\infty}}{\|\theta\|_2}$. This quantity has previously appeared in the matrix completion literature \cite{negahban2012restricted,davenport20141,davenport2016overview} in connection with incoherence conditions.

\begin{theorem}[Partial recovery]\label{maintheorem}
   For every small $\nu>0$, we assume     
    $ \mu(\theta) =\frac{\|\theta\|_{\infty}}{\|\theta\|_2} \lesssim \frac{1}{1+(p/n)^{1/4}}$. If  
    \begin{align} \label{lambdacon1}
         \lambda \ge \|\theta\|_\infty +\sigma\sqrt{2(1+\nu)\log(np)}
    \end{align} 
    and $1\lesssim r_{n,\lambda}\lesssim \min\{(np)^{2\nu/5},n^{2/5}\}$, then for $t\ge \lceil 3\log_4n\rceil$,    we have 
    \begin{gather} 
    \label{highpbound}
       \mathbb{P}\bigg(\ell(\hat{\eta}^t,\eta)\le \exp\Big(-\frac{(1-\epsilon_{n,\lambda})r_{n,\lambda}^2}{2}\Big)\bigg)\ge 1-r_{n,\lambda}^{-1/2}-2(np)^{-\nu}\,,
    \end{gather}
    for some small enough $\epsilon_{n,\lambda}=O\!\left(\frac{1}{r_{n,\lambda}}+\left(\frac{\log n}{n}\right)^{1/4}\right)$.
\end{theorem}
Our result resembles the sharp misclassification rate in \cite{ndaoud2022sharp}, with the noise level $\sigma$ in \eqref{sharppari} replaced by $\lambda$. 
In addition, we require $\mu(\theta)\lesssim \frac{1}{1+(p/n)^{1/4}}$ and $r_{n,\lambda}\lesssim \min\{(np)^{2\nu/5},n^{2/5}\}$. The latter appears to be a removable technical condition and is quite mild, since $r_{n,\lambda}\asymp \sqrt{\log n}$ already leads to exact recovery (see Corollary \ref{cor:exact} below). The following remark further discusses the spikiness condition.
\begin{rem}\label{rem:spiky1}
To see why the spikiness bound is necessary, observe that
\begin{align}
\textrm{$r_{n,\lambda} \gtrsim 1$\, is equivalent to\, $\lambda \lesssim \frac{\|\theta\|_2}{1+(p/n)^{1/4}}$\,,}
\label{snrgtr1}
\end{align}
thus $\mu(\theta)\lesssim \frac{1}{1+(p/n)^{1/4}}$ is required to ensure the existence of $\lambda$ satisfying \eqref{lambdacon1}. In light of $\mu(\theta)\ge \frac{1}{\sqrt{p}}$, a necessary (but not sufficient) condition for our result to hold is $p\gtrsim 1$. We develop a Haar matrix preprocessing step in Section \ref{sec:haar} to ensure that $\mu(\theta)\lesssim \frac{1}{1+(p/n)^{1/4}}$ holds after a random rotation of the data when $p\gtrsim 1$. 
\end{rem}


The next remark shows that dithering is necessary to achieve Theorem \ref{maintheorem}. 
\begin{rem}\label{rem:dither} 
We show that direct quantization without dithering does not yield guarantees comparable to those in existing works \cite{ndaoud2022sharp,fei2018hidden,lu2016statistical,abbe2022lp}. Consider $\theta = \frac{\Delta}{\sqrt{p}}\mathbf{1}_p$ (so that $\|\theta\|_2=\Delta$) and let $\varepsilon_i$ have i.i.d.~Rademacher entries (i.e., $\mathbb{P}(\varepsilon_{ij}=1)=\mathbb{P}(\varepsilon_{ij}=-1)=\frac{1}{2}$, which is encompassed by our setting with $\sigma\asymp 1$). Then, when $\Delta \le \sqrt{p}/2$, we have $\dot{X}_i=\sign(X_i) = \sign(\eta_i\theta+\varepsilon_i) = \sign(\varepsilon_i) = \varepsilon_i$, which are pure noise and contain no information. Thus, in this example, accurate clustering can only occur when $\|\theta\|_2=\Delta \ge \sqrt{p}/2$. In contrast, Theorem \ref{maintheorem} only requires $\|\theta\|_2\gtrsim (1+(p/n)^{1/4})\lambda$ which is substantially weaker in high dimensions.  
\end{rem}

The proof of Theorem \ref{maintheorem} builds on \cite{ndaoud2022sharp} but requires additional work to handle the quantization and introduces several new ideas to obtain a sharp constant in the exponent. The main difficulty is that the quantization breaks the rotational invariance of the Gaussian noise exploited in \cite{ndaoud2022sharp}. In Section \ref{sec:outlineproof}, we provide an overview of the proof with an emphasis on these additional technical aspects. The complete proof is given in Appendix \ref{app:proofthm1}.

Since $\ell(\hat{\eta},\eta)\in\{\frac{i}{n}:i\in[n]\}$, $\ell(\hat{\eta},\eta)<\frac{1}{n}$ implies $\hat{\eta}=\pm\eta$, namely exact recovery of the labels. Therefore, the partial recovery rate in Theorem \ref{maintheorem} yields the following result. Its proof is given in Appendix \ref{app:proveexact1}.

\begin{corollary}
    [Exact recovery]\label{cor:exact} In the setting of   Theorem \ref{maintheorem} with a stronger spikiness condition
 $ \mu^2(\theta) \le \frac{0.99}{\log n+\sqrt{\log^2n+\frac{2p\log n}{n}}}$, assume that the separation condition 
    \begin{align}\label{sepaexact}
        \|\theta\|_2^2 \ge (1+\epsilon)\lambda^2 \bigg(1+\sqrt{1+\frac{2p}{n\log n}}\bigg)\log n
    \end{align}
    holds for some $\epsilon=O(\frac{1}{\sqrt{\log n}})$, then  for any $t\ge \lceil3\log_4n\rceil$, $\hat{\eta}^{t}=\pm \eta$ holds with probability at least $1-O(\frac{1}{\sqrt{\log n}})-2(np)^{-\nu}$. 
\end{corollary}

In two-component Gaussian mixture models without quantization, it was established that Lloyd's algorithm \cite{ndaoud2022sharp} and semidefinite programming \cite{chen2021cutoff} attain exact recovery under the separation condition \eqref{existsepara}. The separation condition in Corollary \ref{cor:exact} takes a similar form, with $\sigma$ replaced by a slightly larger $\lambda$ (see \eqref{lambdacon1}).

\begin{rem} 
\label{rem:spiky2}
We impose the spikiness bound $\mu^2(\theta) \le \frac{0.99}{\log n+\sqrt{\log^2 n+\frac{2p\log n}{n}}}$ to ensure that \eqref{lambdacon1} and \eqref{sepaexact} hold simultaneously under $\lambda = \|\theta\|_{\infty}+\sigma\sqrt{2(1+\nu)\log(np)}$ and $\|\theta\|_2^2\gtrsim \sigma^2 \log(np)(\log n + \sqrt{\frac{p\log n}{n}})$. Since $\mu^2(\theta)\ge \frac{1}{p}$, the condition $p\gtrsim \log n$ is necessary, though not sufficient, for the spikiness condition to hold. We will see that an additional Haar matrix preprocessing step can help bypass the spikiness condition whenever $p\gtrsim \log n \log\log n$. 
\end{rem}

\begin{rem}
\label{rem:exact}
Estimation under one-bit quantization has been an active research area, but existing results typically provide only approximate recovery guarantees \cite{jacques2013robust,davenport20141,dirksen2022covariance,cai2022distributed}. In contrast, exact recovery is achievable in clustering under one-bit quantization, due to the fact that the desired label vector $\eta$ takes discrete values in $\{-1,1\}^n$. Technically, while existing works control the impact of quantization on responses \cite{dirksen2021non,thrampoulidis2020generalized} and covariance matrices \cite{dirksen2021non,chen2025parameter}, our analysis examines the interplay  between the quantization and the hollowed Gram matrix.    
\end{rem}

\subsection{Recovery Guarantees under Haar Matrix Preprocessing}\label{sec:haar}
A limitation of Theorem \ref{maintheorem} and Corollary \ref{cor:exact} is the spikiness condition on $\theta$, which holds only when $p\gtrsim 1$ for partial recovery and $p\gtrsim \log n$ for exact recovery. The issue is that, in clustering, the center $\theta$ is unknown (indeed, estimating $\theta$ is itself an important problem \cite{wu2021randomly}), so one cannot verify the spikiness condition before applying our one-bit clustering method. Note that even when $p$ is sufficiently large, $\mu(\theta)$ can take any value in $[\frac{1}{\sqrt{p}},1]$.



We address this issue via a simple Haar matrix preprocessing step prior to one-bit quantization: we draw a Haar matrix $R\sim {\rm Haar}(\mathbb{O}(p))$\footnote{This means that $R$ is uniformly distributed over the group of $p\times p$ orthonormal matrices.} and transform the original samples $\{X_i\}_{i=1}^n$ to $\{RX_i\}_{i=1}^n$. The subsequent quantization and algorithm remain unchanged. In practice, the generation and communication of $R$ can be controlled by a random seed and incur negligible cost (e.g., \cite{shrivastava2024sketching}).

The Haar matrix resolves this issue by \emph{whitening the center}: in light of 
$$RX_i = \eta_i(R\theta)+R\varepsilon_i\,,$$
the original deterministic center $\theta$ is transformed into $R\theta$, which is uniformly distributed over $\|\theta\|_2 S^{p-1}$. A standard argument shows that $\mu(R\theta)\le \sqrt{3\log(p)/p}$ holds with high probability (cf.\ Lemma \ref{Rthetamu}), and therefore the spikiness conditions in Theorem 
\ref{maintheorem} and Corollary \ref{cor:exact} 
are automatically satisfied in sufficiently high dimension:
\begin{itemize}[leftmargin=5ex,topsep=0.25ex]
    \setlength\itemsep{-0.3em}
    \item Partial recovery is attained when $p\gtrsim 1$, since this guarantees $\sqrt{\frac{\log p}{p}}\lesssim \frac{1}{1+(p/n)^{1/4}}$ (see Theorem \ref{thm:spifree});
    \item Exact recovery is attained when $p \gtrsim \log n\log\log n$, since this guarantees $\sqrt{\frac{\log p}{p}}\lesssim\frac{1}{\sqrt{\log n+\frac{p\log n}{n}}}$ (see Corollary \ref{coro:haar}).  
\end{itemize}
The following theorem concerns the misclassification rate under the Haar matrix preprocessing. 


\begin{theorem}[Partial recovery without spikiness condition] \label{thm:spifree}
Suppose that $R\sim {\rm Haar}(\mathbb{O}(p))$ is independent of everything else, and in our model we quantize  $X_i$  to   $\dot{X}_i=\sign(RX_i+\tau_i)$ with $\tau_i\sim {\rm Unif}[-\lambda,\lambda]^p$, then we compute $\{\hat{\eta}^k\}_{k\ge 0}$ by Algorithm \ref{alg:2lloyd}.  Given small $\nu>0$, assume $\min\{n,p,r_{n,\lambda}\}\gtrsim 1$ and $r_{n,\lambda}\lesssim \min\{(np)^{2\nu/5},n^{2/5}\}$, where recall $r_{n,\lambda}= \frac{\|\theta\|_2^2/\lambda^2}{\sqrt{\|\theta\|_2^2/\lambda^2+p/n}}$. If 
\begin{align}\label{lambdahaar}
    \lambda\ge \sqrt{\frac{3\log p}{p}}\|\theta\|_2 + \sigma\sqrt{2(1+\nu)\log(np)}\,,
\end{align}
then for any $t\ge \lceil3\log_4n\rceil$ we have 
\begin{align*}
    \mathbb{P}\bigg(\ell(\hat{\eta}^t,\eta)\le \exp\bigg(-\frac{(1-\epsilon_{n,\lambda})r_{n,\lambda}^2}{2}\bigg)\bigg)\ge 1-r_{n,\lambda}^{-1/2}-2(np)^{-\nu} - 3p^{-1/4} 
\end{align*}
for some $\epsilon_{n,\lambda}=O(r_{n,\lambda}^{-1}+(\frac{\log n}{n})^{1/4}+\log^{-1/4}(np))$.  
\end{theorem}
Overall, the result is proved by revisiting the arguments for Theorem \ref{maintheorem}. One aspect that requires adaptation is that, when extending to general sub-Gaussian noise, $R\varepsilon_i$ may have correlated entries. We address this issue using a noise decomposition and conditioning argument (cf. Remark \ref{rem:proofthm2} in the appendix). The complete proof appears in Appendix \ref{app:proofthm2}.

Enforcing $\ell(\hat{\eta}^t,\eta)<\frac{1}{n}$ yields the following exact recovery guarantee, where we choose the minimal $\lambda$ in \eqref{lambdahaar} for simplicity. See Appendix \ref{app:provecoro2} for the proof.

\begin{corollary}[Exact recovery without spikiness condition] \label{coro:haar}
In the setting of Theorem \ref{thm:spifree} with
\begin{align}
    \lambda = \sqrt{\frac{3\log p}{p}}\|\theta\|_2+\sigma\sqrt{2(1+\nu)\log(np)},\label{minimallambda}
\end{align}
if for some $\epsilon=O(\frac{1}{\log^{1/4}n})$ the separation condition
\begin{align} \label{separation111}
    \|\theta\|_2^2 \ge (1+\epsilon)\bigg(\sqrt{\frac{3\log p}{p}}\|\theta\|_2+\sigma\sqrt{2(1+\nu)\log(np)}\bigg)^2\bigg(1+\sqrt{1+\frac{2p}{n\log n}}\bigg)\log n
\end{align}
holds, then for any $t\ge \lceil 3\log_4 n\rceil$, $\hat{\eta}^t=\pm \eta$ with probability at least $1-O(\frac{1}{\log^{1/4}n}+(np)^{-\nu}+p^{-1/4})$. More specifically, if $\min\{n,p,r_{n,\lambda}\}\ge C_1(\nu)$ and $p\ge C_2(\nu)\cdot \log n \log\log n$ for sufficiently large constants $C_1(\nu),C_2(\nu)$ depending only on $\nu$, then the separation condition in \eqref{separation111} can be ensured by the more explicit condition
\begin{align}\label{oursepara}
    \|\theta\|_2^2 \ge 2 (1+2\nu)\sigma^2\left(1+\sqrt{1+\frac{2p}{n\log n}}\right)\log(np)\log n.
\end{align}
\end{corollary}
The above result states that exact recovery is achieved under \eqref{oursepara} provided that $p\gtrsim_{\nu} \log n\log\log n$. The separation condition \eqref{oursepara} is explicit and differs from the sharp separation condition \eqref{existsepara} for two-component Gaussian mixture models without quantization by only an additional factor of $2\log(np)$. 
 \begin{rem}\label{rem:tuning}
   The tuning of $\lambda$ appears to require knowledge of $(\|\theta\|_2,\sigma)$ in view of \eqref{minimallambda}, but we note that in many regimes it suffices to have an estimate on the noise level $\sigma$. In particular, if we assume $\|\theta\|_2\le c\sigma\sqrt{\frac{p\log(np)}{\log p}}$ for sufficiently small $c$,\footnote{This is mild because under $p\gtrsim \log n\log\log n$, the upper bound $\sigma\sqrt{\frac{p\log(np)}{\log p}}$ is much larger than $\|\theta\|_2 \gtrsim \sigma (\log^2(np)\log^2 n+\frac{p\log^2(np)\log n}{n})^{1/4}$ required in \eqref{oursepara}.} then the term $\sqrt{\frac{3\log p}{p}}\|\theta\|_2$ is negligible compared to $\sigma\sqrt{2\log(np)}$, and hence the theory suggests $\lambda\approx \sigma\sqrt{2\log(np)}$. In practice, we recommend $\lambda = s\sigma\sqrt{2\log(np)}$ for some shrinkage parameter $s\in(0,1)$; see, e.g., the second experiment in Section \ref{sec:experi}. 
   We leave a more thorough investigation of the tuning of $\lambda$ for future work. 
 \end{rem}
\begin{rem}
To extend our method to $p\lesssim 1$, a naive approach is to append additional zeros to the samples. That said, numerical results suggest that our method already performs well and achieves exact recovery in a dimension as low as $p=5$; see Figure \ref{fig:sub1} in Section \ref{sec:experi} for instance. 
\label{rem:pO1}
\end{rem}
 
\subsection{Technical Overview (Theorem \ref{maintheorem})}\label{sec:outlineproof}


{\bf (A) Handing the quantization.} Since  in our one-bit setting $\lambda\dot{X}_i$ serves as a surrogate for $X_i$, a useful perspective is to view it as a clustering problem for $\{\lambda\dot{X}_i\}_{i=1}^n$, with the 
 overall noise on the $i$-th sample being \[u_i=\lambda\dot{X}_i-\eta_i\theta= \lambda\sign(X_i+\tau_i)-\eta_i\theta=\lambda\sign(\eta_i\theta+\varepsilon_i+\tau_i)-\eta_i\theta.\] One difficulty here is that $u_i$ is not mean-zero (a property that existing analyses \cite{lu2016statistical,ndaoud2022sharp,giraud2019partial}  heavily rely on). 
 We bypass this hurdle via a proof scheme \cite{chen2025parameter} that draws connection between the one-bit quantizer and the uniform quantizer, consisting of three steps:

 (i) Define the uniform quantizer with resolution $\delta>0$ as $Q_\delta(a)=\delta(\lfloor\frac{a}{\delta}\rfloor+\frac{1}{2})$, and note that the (rescaled) one-bit quantizer and $Q_{2\lambda}$ are connected by (cf. Figure \ref{fig:Q2_sign} in the appendix)
 \begin{align}\nonumber
      &Q_{2\lambda}(a) = \lambda\sign(a),\quad \forall|a|\le 2\lambda\\\nonumber
      \Longrightarrow~& Q_{2\lambda}(a+\tau)=\lambda\sign(a+\tau)\,\textrm{~under $\tau\sim {\rm Unif}[-\lambda,\lambda]$}, \quad \forall |a|\le \lambda\\
         \Longrightarrow~&   \lambda\dot{X}_i = \underbrace{Q_{2\lambda}(X_i+\tau_i)}_{\tilde{X}_i},\quad \forall i\in[n], \quad\text{if}\,~\|Y\|_{\max} = \max_{i\in[n]}\|X_i\|_{\infty} \le \lambda. \label{goack}
 \end{align}
(ii) We then analyze the clustering of $\{\tilde{X}_i:=Q_{2\lambda}(X_i+\tau_i)\}_{i=1}^n$ with  overall noise $\tilde{u}_i=\tilde{X}_i-\eta_i\theta$ having independent, \emph{zero-mean}, sub-Gaussian entries (cf. Lemma \ref{lem:quannoi});

(iii) By the sub-Gaussian tail bound   $\lambda\ge \|\theta\|_\infty + L\sqrt{2(1+\nu)\log(np)}$ ensures that $\|Y\|_\infty\le \lambda$ holds with high probability (w.h.p.), which together with (\ref{goack}) allows us to carry over the conclusion for $\{\tilde{X}_i\}_{i=1}^n$ to $\{\lambda\dot{X}_i\}_{i=1}^n$.


{\bf (B) Extending the argument in \cite{ndaoud2022sharp}.} We focus on Step (ii). We analyze the performance of Algorithm \ref{alg:2lloyd} with $\tilde{Y}=[\tilde{X}_1,\cdots,\tilde{X}_n]$ rather than $\dot{Y}$. Note that the zero-mean $\tilde{u}_i=\tilde{X}_i-\eta_i\theta$ allows us to closely follow the analysis in \cite{ndaoud2022sharp} which (a) first shows a crude bound on the iterates $\ell(\hat{\eta}^k,\eta)=O(n^{-2}+r_{n,\lambda}^{-2})$ and (b) then tightens it to  
\begin{align*}
    \ell(\hat{\eta}^k,\eta) \lesssim \frac{1}{n^{3/2}} + r_{n,\lambda}^{5/2}P^*~~{\rm w.h.p.},\quad (\forall k\ge \lceil 3\log_4n\rceil),
\end{align*}
where, with $\delta_j=Q_{2\lambda}(\theta+\varepsilon_j+\tau_j)-\theta$ for $j\in[n]$, $P^*$ is defined as 
\begin{align}\label{definePstar11}
    P^*:= \mathbb{P} \bigg(\Big\langle \theta+\delta_1,\theta+\frac{1}{n-1}\sum_{j=2}^n \delta_j\Big\rangle < \frac{C\|\theta\|_2^2}{r_{n,\lambda}}\bigg).
\end{align}
 This part of the  argument is presented in Step 2 in Appendix \ref{app:proofthm1} and only involves some straightforward extensions of the concentration bounds in \cite{ndaoud2022sharp} (cf. Lemmas \ref{lem:sgvbound}--\ref{utubound}).

 {\bf (C) Establishing sharp bound on $P^*$.}  
It remains to establish a sharp bound on $P^*$ (accomplished in  Step 4 of  Appendix \ref{app:proofthm1}). While the corresponding analysis in \cite[Thm. 5]{ndaoud2022sharp} deals with  standard Gaussian vector and hence   rotational invariance readily yields sharp tail bound for the marginals, the $\delta_j$'s in Equation (\ref{definePstar11}) are high-dimensional \emph{non-Gaussian} random vectors and in turn the analysis here is much more entangled.
We begin with \begin{align}
        \bigg\langle\theta+\delta_1,\theta+\frac{1}{n-1}\sum_{j=2}^n\delta_j\bigg\rangle=\|\theta\|_2^2+ \bigg\langle\theta,\frac{1}{n-1}\sum_{j=2}^n \delta_j\bigg\rangle + \bigg\langle\delta_1,\theta+\frac{1}{n-1}\sum_{j=2}^n\delta_j\bigg\rangle.
\end{align}
In view of $\|\theta\|_2^2$, the term $\frac{C\|\theta\|_2^2}{r_{n,\lambda}}$ in (\ref{definePstar11}) has minimal impact due to $r_{n,\lambda}\gtrsim 1$. By $\|\delta_j\|_{\psi_2}\lesssim \lambda$, the standard sub-Gaussian tail bound shows that $\langle \theta,\frac{1}{n-1}\sum_{j=2}^n\delta_j\rangle$ is negligible. The main bulk of techniques lies in treating $\langle\delta_1,\theta+\frac{1}{n-1}\sum_{j=2}^n\delta_j\rangle$:

{\it (C.1) Establishing sharp bound on $\|\theta+\frac{1}{n-1}\sum_{j=2}^n\delta_j\|_2$.} We start with $
     \big\|\theta+\frac{1}{n-1}\sum_{j=2}^n\delta_j\big\|_2^2  =\|\theta\|_2^2+\big\|\frac{1}{n-1}\sum_{j=2}^n\delta_j\big\|_2^2 + 2\big\langle \theta,\frac{1}{n-1}\sum_{j=2}^n\delta_j\big\rangle.$ A sub-Gaussian tail bound renders $\langle\theta,\frac{1}{n-1}\sum_{j=2}^n\delta_j\rangle$   
negligible. To   bound $\|\frac{1}{n-1}\sum_{j=2}^n\delta_j\|_2^2$, we first use a Bernstein's inequality to show that it sharply concentrates about $\mathbb{E}\|\frac{1}{n-1}\sum_{j=2}^n\delta_j\|_2^2=\frac{1}{n-1}\mathbb{E}\|\delta_2\|_2^2$, and then leverage $Q_{2\lambda}(\theta+\varepsilon_i+\tau_i)\approx  \lambda\sign(\theta+\varepsilon_i+\tau_i)$ to compute the expectation with sharp constant; see Lemma \ref{lem:boundsumvnorm}. We then reach $\|\frac{1}{n-1}\sum_{j=2}^n\delta_j\|_2^2\le(1+o(1))(\|\theta\|_2^2+\frac{p\lambda^2}{n})$. 

{\it (C.2) Passing to the margins of $\delta_1$.} The sharp bound on $\|\theta+\frac{1}{n-1}\sum_{j=2}^n\delta_j\|_2$ readily yields $P^*\le P^{**}+\textrm{negligible terms}$, where  \begin{align*}
    P^{**} = \mathbb{P}\bigg(\langle \delta_1,\hat{\xi}\rangle<-(1-o(1))\frac{\|\theta\|_2^2}{\sqrt{\|\theta\|_2^2+p\lambda^2/n}}\bigg),~\textrm{for}~~\hat{\xi}:=\frac{\theta+(n-1)^{-1}\sum_{j=2}^n\delta_j}{\|\theta+(n-1)^{-1}\sum_{j=2}^n\delta_j\|_2}\in\mathbb{S}^{p-1}.
\end{align*}

 {\it (C.3) Bounding $P^{**}$ sharply.} We    reduce the random vector $\delta_1=Q_{2\lambda}(\theta+\varepsilon_j+\tau_j)-\theta$ to the bounded $\tilde{\delta}_1:=\lambda\sign(\theta+\varepsilon_1+\tau_1)-\mathbb{E}[\lambda\sign(\theta+\varepsilon_1+\tau_1)]$ in Lemma \ref{lem:Pstarstar} and then invoke Hoeffding's inequality, which happens to yield sharp bound in this regime.

\section{Minimax Lower Bound}\label{sec:lower}
In this section, we establish a minimax lower bound for clustering based on the one-bit samples $(\dot{X}_i)_{i=1}^n$, showing that our results are sharp in general. Note that existing lower bounds (see \cite{ndaoud2022sharp,chen2021cutoff,lu2016statistical}) are derived under isotropic Gaussian noise $\varepsilon_i\sim N(0,\sigma^2I_p)$, and since the quantization breaks rotational invariance, our proof requires several different ideas. In particular, the main lemma is the anti-concentration bound for a binomial variable (cf. Lemma \ref{lem:anti}).


\begin{theorem}[Minimax lower bound]\label{thmlower}
    Given $\|\theta\|_2=\Delta$, assume that $(\varepsilon_i)_{i=1}^n$ are i.i.d. $N(0,\sigma^2I_p)$ vectors, $\lambda\gtrsim\frac{\Delta}{\sqrt{p}}+\sigma$, and that a mild scaling condition $\log(\frac{2\lambda\sqrt{p}}{\Delta})\le \frac{\lambda^2}{8\sigma^2}$ holds. Then let $\hat{\eta}=\hat{\eta}(\{\dot{X}_i\}_{i=1}^n)$ be a measurable function of $\{\dot{X}_i=\sign(X_i+\tau_i)\}_{i=1}^n$ (where $\tau_i\stackrel{iid}{\sim}{\rm Unif}[-\lambda,\lambda]^p$),
    we have  
    \begin{align*}
        \inf_{\hat{\eta}}\sup_{\theta,\eta}\, \mathbb{E} \big[\ell(\hat{\eta},\eta)\big]\ge \frac{1}{\sqrt{2p}}\exp\Big(-\frac{(1+\epsilon)\Delta^2}{2\lambda^2}\Big)
    \end{align*}
  where the constant $\epsilon\to0$ when $\frac{\lambda}{p^{-1/2}\|\theta\|_2 +\sigma}\to \infty$. 
\end{theorem} 
We now discuss in what regimes the above lower bound implies the sharpness of our recovery guarantees in Section \ref{sec:main}. 
\begin{rem} \label{rem:whensharp}
    By absorbing the leading factor $\frac{1}{\sqrt{2p}}$ into the exponent, the lower bound reads $
          \frac{1}{\sqrt{2p}}\exp\big(-\frac{(1+\epsilon)\|\theta\|_2^2}{2\lambda^2}\big) = \exp\big(-(1+\epsilon+\frac{\log(2p)\lambda^2}{\|\theta\|_2^2})\frac{\|\theta\|_2^2}{2\lambda^2}\big)$ 
     and can be written as $\exp\big(-(1+\epsilon')\frac{\|\theta\|_2^2}{2\lambda^2}\big)$ for small enough $\epsilon'$ under $  \|\theta\|_2^2    \gtrsim  \lambda^2\log(2p)$. Recall that our misclassification rates in Theorems \ref{maintheorem}, \ref{thm:spifree} read $\exp\big(-\frac{1-o(1)}{2}r_{n,\lambda^2}\big)=\exp\big(-\frac{1-o(1)}{2}\frac{(\|\theta\|_2/\lambda)^4}{(\|\theta\|_2/\lambda)^2+p/n}\big)$. Hence, it reduces to $\exp\big(-\frac{1\pm o(1)}{2}\frac{\|\theta\|_2^2}{\lambda^2}\big)$ under $p\ll \frac{n\|\theta\|_2^2}{\lambda^2}$.\footnote{We write this to denote $p\le \frac{cn\|\theta\|_2^2}{\lambda^2}$ for some small enough $c$.} In conclusion, under the two scaling conditions of 
     \begin{align}
         \label{matchregime}
         \textrm{$\|\theta\|_2^2\gtrsim \lambda^2\log(2p)$ ~~and~~ $p\ll\frac{n\|\theta\|_2^2}{\lambda^2}$,}
     \end{align}
     the lower bound matches the upper bounds in Theorems \ref{maintheorem} and \ref{thm:spifree}, indicating  that $\exp\big(-\frac{1\pm o(1)}{2}\frac{\|\theta\|_2^2}{\lambda^2}\big)$ is the sharp misclassification rate. Since exact recovery is achieved if and only if $\ell(\hat{\eta}_t,\eta)<\frac{1}{n}$, the lower bound also indicates the sharpness of the separation conditions (\ref{cor:exact}) and (\ref{oursepara}) in the regime of (\ref{matchregime}).
\end{rem}
\begin{rem}
   However, we do not have a matching lower bound in some regimes of interest, such as the high-dimensional regime $p\gg \frac{n\|\theta\|_2^2}{\lambda^2}$, where our upper bound reads $\exp\big(-\frac{1-o(1)}{2}\frac{n\|\theta\|_2^4}{p\lambda^4}\big)$. We leave the development of tighter lower and upper bounds in these regimes to future work. \label{rem:limitlower}
\end{rem}

We provide an overview of the proof here and relegate the complete proof to Appendix \ref{app:prooflower}.  

\begin{proof}[Overview of the Proof for Theorem \ref{thmlower}] 
Unlike in \cite{lu2016statistical,ndaoud2022sharp}, under quantization,
the Bayes classifier does not have a closed form under general $\theta$. Our remedy is to fix the specific center as $\theta=\frac{\Delta}{\sqrt{p}}\mathbf{1}_p$. We then place an independent Rademacher prior on the label $\eta$ and leverage a standard argument (e.g., \cite{gao2018community}) to get  
$
       \inf_{\hat{\eta}}\sup_{\theta,\eta}\, \mathbb{E} \big[\ell(\hat{\eta},\eta)\big]  \ge c \cdot \inf_{\bar{\eta}_j}\,\mathbb{E}_\pi \mathbb{E}\,|\bar{\eta}_j(\dot{X}_j)-\eta_j|,$
where   $\eta_j$ is a Rademacher variable, and $\bar{\eta}_j$ is a measurable function of $\dot{X}_j$ for estimating $\eta_i$, $c>0$ is a universal constant. 
Since the entries of $\theta$ are equal and the noise and dithering are symmetric, we find that $\eta_j=1$ renders higher likelihood if and only if the number of $1$'s in $\dot{X}_j$ is higher than the number of $-1$'s. Thus, the Bayes optimal selector attaining the infimum $\inf_{\bar{\eta}_j}\,\mathbb{E}_\pi \mathbb{E}\,|\bar{\eta}_j(\dot{X}_j)-\eta_j|$ is given by $\eta_j^* = \sign(\langle \dot{X}_j,\mathbf{1}\rangle)$. Note that it makes an error if and only if {\it at least half  of the  entries of $\dot{X}_j$ have different signs from $\eta_j$.} 
  Observe that the entries of $\dot{X}_j$ are independent and have sign different from $\eta_j$ with the same probability $P_-:=\mathbb{P}(\frac{\|\theta\|_2}{\sqrt{p}}+\varepsilon_{ij}+\tau_{ij}<0)$. Therefore, the number of $-\eta_j$ in $\dot{X}_j$  follows ${\rm Binomial}(p,P_-)$, and   thus by Lemma \ref{lem:anti} 
$\mathbb{E}_\pi \mathbb{E}\,|\bar{\eta}_j(\dot{X}_j)-\eta_j|=\mathbb{E}|\eta_j^*-\eta_j|=2 \mathbb{P}\Big({\rm Binomial}(p,P_-)>\frac{p}{2}\Big)\ge \frac{1}{\sqrt{2p}}\exp\Big(-pD_{\rm KL}\big(\,\frac{1}{2}\,\| P_-\big)\Big).$ 
By further computing $P_-$ and performing some algebra, we obtain the claimed lower bound under the scaling assumptions of the theorem.
\end{proof}

\section{Experiments}\label{sec:experi}
We provide experimental results to corroborate our theory and demonstrate the effectiveness of the proposed method. We consider two Gaussian mixtures with $\theta$ uniformly distributed over $\Delta \mathbb{S}^{p-1}$ and  $\varepsilon_i\sim N(0,I_p)$.\footnote{Due to  $\theta\sim {\rm Unif}(\Delta\mathbb{S}^{p-1})$ we do not need a separate Haar matrix rotation step.} We run $200$ iterations in Algorithm \ref{alg:2lloyd} and the reported results are averaged over $300$ independent trials. All experiments were implemented using Matlab R2022a on a laptop with an Intel CPU up to 2.5 GHz and 32 GB RAM. We defer some details to Appendix \ref{app:detailexp} due to page limit.

\noindent\textbf{Phase Transitions.} The aim of our first two experiments is to illustrate the separation condition (\ref{sepaexact}) for achieving exact recovery. Note that  the minimal $\lambda$ in (\ref{lambdacon1}) is approximately $\sqrt{2\log(np)}$ (cf. Equation (\ref{minimallambda}), Remark \ref{rem:tuning}). Further illustrations of the experimental designs can be found in Appendix \ref{app:details}.

We start with a low-dimensional setting with   $p=5$, in which (\ref{sepaexact}) is provably sharp (cf. Remark \ref{rem:whensharp}). We set the dithering level $\lambda=\sqrt{2\log n}$ in  light of $\sqrt{2\log(np)}\approx \sqrt{2\log n}$ under $p=5$. In this setting, the separation condition (\ref{sepaexact}) approximately reduces to $\Delta\ge 2\log n$. We set $a=\log n$, $b= \Delta$ and test  $a=3.4:0.25:6.9,\,b=4:0.25:16$. We then report the empirical rates of exact recovery in Figure \ref{fig:sub1}. Consistent with our theory, the results suggest that the phase transitions from partial recovery to exact recovery occur roughly in $b=2a$.

Fix $n=100$, we also follow \cite{ndaoud2022sharp} to provide a high-dimensional setting with $p=bn\log n$, $\Delta^2 = \lambda^2(1+\sqrt{a})\log n$, and $\lambda=\sqrt{2\log (np)}$. Hence, our (\ref{sepaexact}) suggests that $a \ge 1+2b$ ensures exact recovery (w.h.p.).  We test $b= 2.4:0.05:4.6$ and $a= 1.5:0.1:13$ and similarly report the empirical exact recovery rates in    Figure \ref{fig:sub2}. While we do not have a matching lower bound in this regime (cf. Remark \ref{rem:limitlower}), it appears that     $a \ge 1+2b$ remains close to the locations of the phase transitions.

\noindent\textbf{One-bit v.s. Classical.} Our third experiment confirms the efficacy of the proposed method by comparing the performance of \cite{ndaoud2022sharp} (using $\{X_i\}_{i=1}^n$) and Algorithm \ref{alg:2lloyd} (using $\{\dot{X}_i\}_{i=1}^n$).   We set $\lambda=0.3\cdot\sqrt{2\log (np)}$ for our method with a shrinkage factor of $0.3$; see further justification in Appendix \ref{app:shrink}.  We test $n=100$, $p=300$ and $\Delta = 0.1:0.1:6$ and report the   misclassification rates in Figure \ref{fig:sub3}. Empirically, our one-bit clustering method requires {\it no more than twice the separation} $\Delta$ to achieve the same misclassification rate as in the classical clustering setting.   

\begin{figure}[ht!]
    \centering

    \begin{subfigure}[t]{0.32\textwidth}
        \centering
        \includegraphics[width=\textwidth]{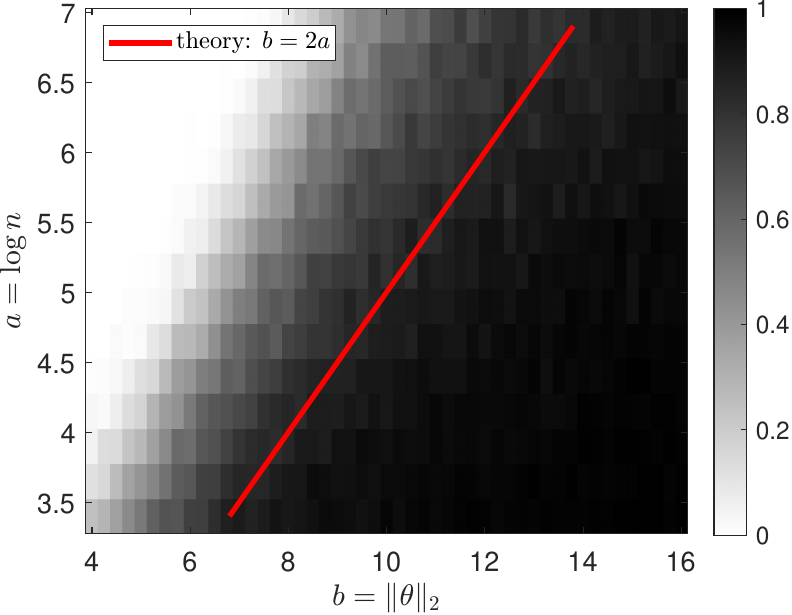}
        \caption{Phase transition  ($p=5$)}
        \label{fig:sub1}
    \end{subfigure}
    \hfill
    \begin{subfigure}[t]{0.31\textwidth}
        \centering
        \includegraphics[width=\textwidth]{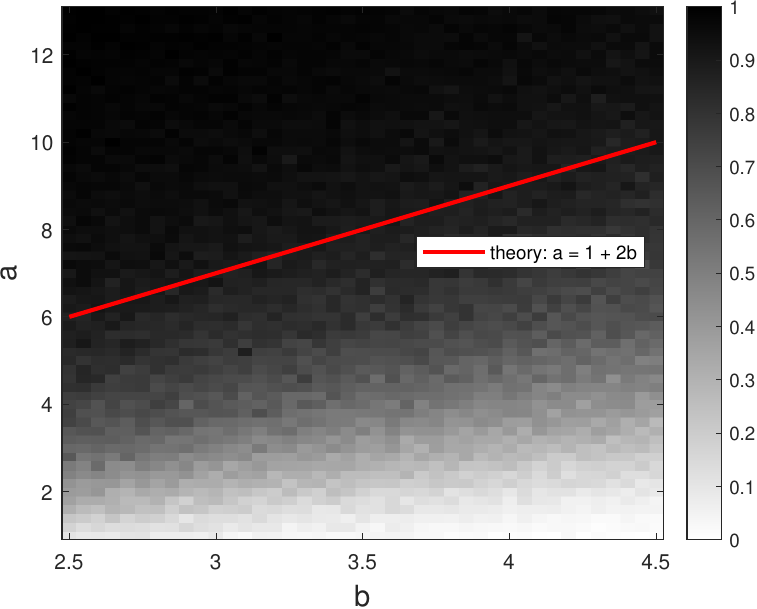}
        \caption{Phase transition  ($p=bn\log n$)}
        \label{fig:sub2}
    \end{subfigure}
    \hfill
    \begin{subfigure}[t]{0.31\textwidth}
        \centering
        \includegraphics[width=\textwidth]{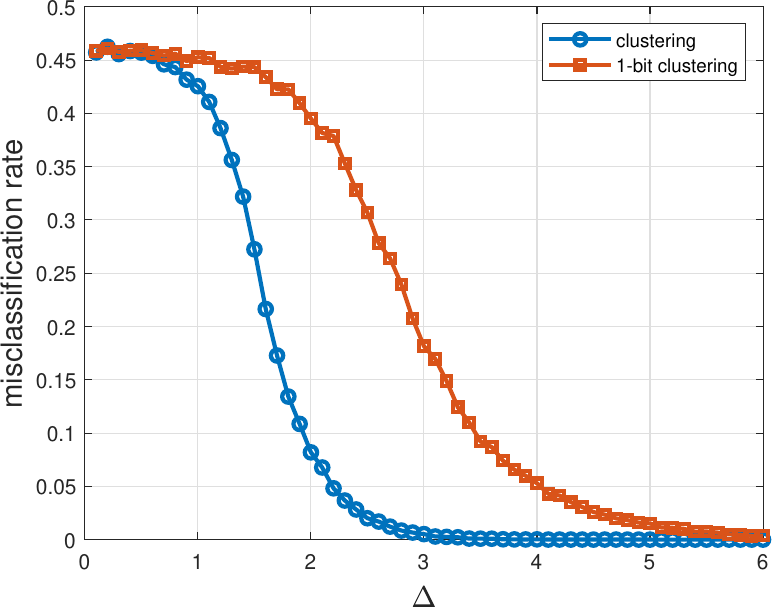}
        \caption{1-bit clustering v.s. clustering}
        \label{fig:sub3}
    \end{subfigure}

    \caption{Phase transitions and comparison with classical clustering.}
    \label{fig:three_figs}
\end{figure}

\section{Conclusion}\label{sec:conclu}
This paper provides the first one-bit clustering method for two symmetric sub-Gaussian mixtures. We adopt a dithered one-bit quantizer and establish partial and exact recovery guarantees that are only slightly worse than those in the classical unquantized setting. 
A lower bound is also provided to show that the recovery guarantees exhibit sharp constants in general. 
Compared to previous work on one-bit estimation, several new phenomena arise, including the achievability of exact recovery and a subtle interaction with the spikiness of the center. The latter can be addressed via a Haar matrix rotation step, which enables partial recovery under $p\gtrsim 1$ and exact recovery under $p\gtrsim \log n\log\log n$. There remain many interesting directions for future work, including tuning-free one-bit clustering (cf.\ Remark \ref{rem:tuning}), sharp bounds in high dimensions (cf.\ Remark \ref{rem:limitlower}),  the performance of other clustering algorithms under quantization, and extension to more than two mixtures that are not in symmetric positions.

\bibliography{libr}

\appendix

 \section{Proof of Theorem \ref{maintheorem} (Partial recovery)} \label{app:proofthm1}
 \begin{proof}
As $\lambda \dot{X}_i$ is the surrogate of $X_i$,
we define the overall noise on the sample as 
\begin{align*}
    u_i =  \lambda \dot{X}_i - \eta_i\theta = \lambda\sign(\eta_i\theta + \varepsilon_i + \tau_i)-\eta_i\theta
\end{align*}
and can formulate the problem in matrix form:
\begin{align*}
    \lambda\dot{Y} = \theta \eta^T + U 
\end{align*}
where $U = [u_1,\cdots,u_n]$. Therefore, the one-bit clustering problem can be treated as a clustering problem with noise $U$. While $u_i$ is bounded and therefore sub-Gaussian, the difficulty is that $u_i$ is not zero-mean. 
We bypass the difficulty by a technical trick developed in \cite{chen2025parameter}: first, we analyze the uniform quantizer that leads to zero-mean noise; second, we show that, when $\lambda\ge \|\theta\|_\infty + L\sqrt{2\log(np)}$, then this uniform quantizer is identical to our one-bit quantizer with high probability. 
\subsubsection*{Step 1: Introducing the Uniform Quantizer}
The uniform quantizer is given by 
\begin{align*}
    Q_{2\lambda}(a) = 2\lambda \Big(\lfloor\frac{a}{2\lambda}\rfloor + \frac{1}{2}\Big),
\end{align*}
and notice that  
\begin{align}\label{equationQ2}
    Q_{2\lambda}(a) = \lambda\sign(a)\,,\quad\forall  |a|\le 2\lambda. 
\end{align}
See the following Figure \ref{fig:Q2_sign}. 

  \begin{figure}[ht!]
\hspace{-1.5cm}\begin{center}
     \includegraphics[height=0.45\textwidth]{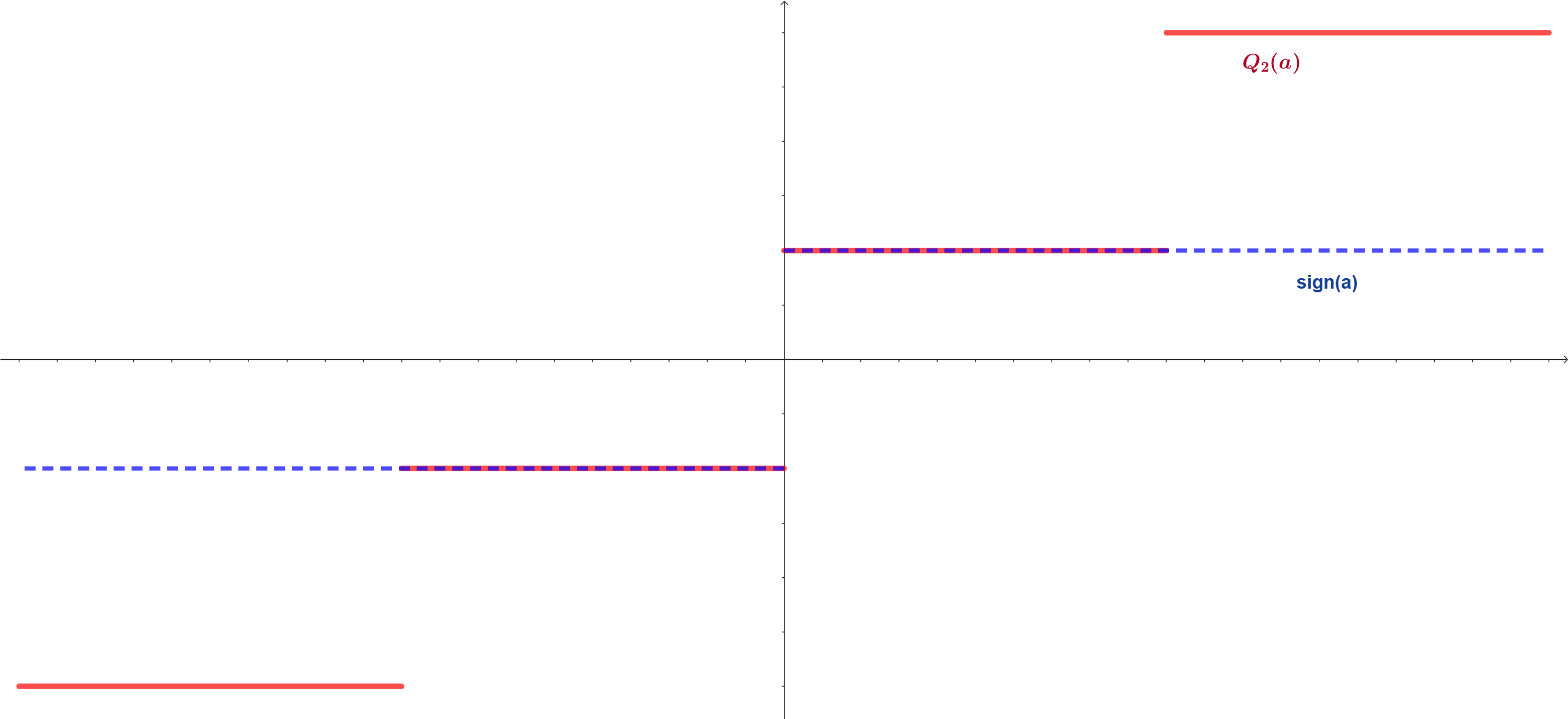}
\end{center}
     \caption{Graphical illustration of Equation (\ref{equationQ2}): the red curve $Q_{2}(a)$ and the blue curve $\sign(a)$ are identical when $|a|<2$.}
     \label{fig:Q2_sign}
 \end{figure}

We now proceed the analysis with the quantized samples $$\tilde{X}_i = Q_{2\lambda}(X_i + \tau_i),\quad i=1,\cdots,n,$$ which can be arranged in the matrix $\tilde{Y}=[\tilde{X}_1,\cdots,\tilde{X}_n]$. We then proceed to analyze the following two-stage procedure, replacing $\dot{Y}$ in Algorithm \ref{alg:2lloyd} by $\tilde{Y}$:
\begin{enumerate}
[leftmargin=5ex,topsep=0.25ex]
    \setlength\itemsep{-0.3em}
    \item Let $\tilde{v}$ be the leading eigenvector of $H(\tilde{Y}^T\tilde{Y})$ and let $\tilde{\eta}^0 = \sign(\tilde{v})$;
    \item Run $\tilde{\eta}^{k+1}=\sign(H(\tilde{Y}^T\tilde{Y})\tilde{\eta}^k)$ to obtain the sequence $\{\tilde{\eta}^i\}_{i\ge 0}$.
\end{enumerate}
We let $e_i=\tilde{X}_i-X_i$ for $i\in [n]$, then we have that $e_i$ is zero-mean and $O(\lambda)$ sub-Gaussian (cf. Lemma \ref{lem:quannoi}). We further let $u_i=e_i+\varepsilon_i$, then the observations $\tilde{X}_i$ can be expressed as
\begin{align*}
    \tilde{X}_i = X_i+e_i = \eta_i\theta + \varepsilon_i + e_i = X_i + u_i,
\end{align*}
or  as the matrix form 
\begin{align*}
    \tilde{Y} = \theta \eta^T + W + E  = \theta\eta^T + U 
\end{align*}
by letting \[\textrm{$W=[\varepsilon_1,...,\varepsilon_n]$,\quad $E=[e_1,...,e_n]$\quad and\quad $U = [u_1,...,u_n]$.}\] Given that the columns of $E$ and $U$ are both zero-mean and sub-Gaussian, we can largely follow the argument in \cite{ndaoud2022sharp} to analyze the algorithm along with a transition from Gaussian noise to sub-Gaussian noise.  

\subsubsection*{Step 2. Extending the Analysis of \cite[Theorems 3 \& 4]{ndaoud2022sharp}}
\subsubsection*{Step 2.1. Crude Bound on Spectral Initialization}   
We use $\tilde{Y}=\theta \eta^T +U$ and start with  
\begin{align*}
    \frac{1}{n}\tilde{Y}^T\tilde{Y} &= \frac{1}{n}(\theta\eta^T+U)^T(\theta\eta^T+U)= \frac{\|\theta\|_2^2}{n}\eta\eta^T + \underbrace{\frac{\eta\theta^TU+U^T\theta\eta^T}{n} + \frac{U^TU}{n}}_{:=Z_1}\,,
\end{align*}
which gives
\begin{align}\label{HYYZ2}
    H\Big(\frac{1}{n}\tilde{Y}^T\tilde{Y}\Big) = \frac{\|\theta\|_2^2}{n}\eta\eta^T  - \frac{\|\theta\|_2^2I_n}{n} +H(Z_1) .
\end{align}
Therefore, by letting $Z_2=H(Z_1)-\frac{\|\theta\|_2^2}{n}I_n$, we reach 
\begin{align}\label{dklem}
    \bigg\|H\Big(\frac{1}{n}\tilde{Y}^T\tilde{Y}\Big)- \frac{\|\theta\|_2^2}{n}\eta\eta^T\bigg\|_{op} \le \|Z_2\|_{op}.
\end{align}
Further noticing that the columns of $U$ are independent and zero-mean,   triangle inequality yields
\begin{align}\nn
    \|Z_2\|_{op} &\le \frac{\|\theta\|^2_2}{n} + 2\bigg\|H\Big(\frac{\eta\theta^TU}{n}\Big)\bigg\|+ \bigg\|H\Big(\frac{U^TU}{n}\Big)\bigg\|_{op}
    \\\nn&\quad\blacktriangleright \textrm{by $Z_2= H\Big(\frac{\eta\theta^TU}{n}\Big)+H\Big(\frac{U^T\theta\eta^T}{n}\Big)+H\Big(\frac{U^TU}{n}\Big)-\frac{\|\theta\|_2^2}{n}I_n$}
    \\\nn 
    &\le \frac{\|\theta\|^2_2}{n} + 4\bigg\|\frac{\eta\theta^TU}{n}\bigg\|_{op} + \bigg\|H\Big(\frac{U^TU}{n}-\mathbb{E}\frac{U^TU}{n}\Big)\bigg\|_{op}
    \\\explain\textrm{by Lemmas \ref{Hbound2} and \ref{lemrow}}
    \\
    &\le \frac{\|\theta\|^2_2}{n} + \frac{4\|\eta\|_2\|U^T\theta\|_2}{n} + \frac{2\|U^TU-\mathbb{E}(U^TU)\|_{op}}{n}. \label{z2divide}\\
    \explain\textrm{by $\|\eta\theta^TU\|_{op}\le \|\eta\|_2\|U^T\theta\|_2$}
\end{align} 
Since the columns of $U$ are independent, zero-mean and have sub-Gaussian norm bounded by 
\begin{align}\label{uisgbound}
    &\|u_i\|_{\psi_2}\le \|\varepsilon_i\|_{\psi_2}+\|e_i\|_{\psi_2}\lesssim \sigma + \lambda \lesssim \lambda, \\\explain\textrm{by Assumption \ref{subgaussian}, Lemma \ref{lem:quannoi}, Equation (\ref{lambdacon1})}
\end{align}
Therefore, $U^T\theta\in \mathbb{R}^n$ is a vector with independent entries whose sub-Gaussian norms are bounded by $O(\lambda\|\theta\|_2)$. Thus, we invoke Lemma \ref{lem:sgvbound} to obtain 
\begin{align}\label{Uthetal2}
    \mathbb{P}\Big(\|U^T\theta\|_2 \le C\lambda \|\theta\|_2 \sqrt{n}\Big) \ge 1- 2\exp(-2n), 
\end{align}
and on the high-probability event $\|U^T\theta\|_2\le C\lambda\|\theta\|_2\sqrt{n}$ we have
\begin{align}\label{etaubound}
    \frac{4\|\eta\|_2\|U^T\theta\|_2}{n} = \frac{4\|U^T\theta\|_2}{\sqrt{n}} \le 4C\lambda\|\theta\|_2.
\end{align} 
Moreover, since $U$ has independent rows of $O(\lambda)$ sub-Gaussian norms,\footnote{To see this, notice that $U=W+E$ and both $W$ and $E$ have independent entries of $O(\lambda)$ sub-Gaussian norms.} by Lemma \ref{utubound} we have
\begin{align}
    \mathbb{P}\Big(\|U^TU-\mathbb{E}(U^TU)\|_{op}\le C\cdot \lambda^2n\max\{1,\sqrt{\frac{p}{n}}\}\Big) \ge 1-2\exp(-2n). \label{utuboundeq}
\end{align}
 Therefore,  
 \begin{align}\nn 
       &\bigg\|H\Big(\frac{1}{n}\tilde{Y}^T\tilde{Y}\Big)- \frac{\|\theta\|_2^2}{n}\eta\eta^T\bigg\|_{op}\le \|Z_2\|_{op} \\\nn&\le \frac{\|\theta\|_2^2}{n} + C\Big(\lambda\|\theta\|_2 +  \lambda^2\max\{1,\sqrt{p/n}\} \Big)\\\explain \textrm{by substituting (\ref{etaubound}), (\ref{utuboundeq}) into (\ref{z2divide})}\\&\le \frac{\|\theta\|_2^2}{n} + C'\Big(\lambda\|\theta\|_2+\lambda^2\sqrt{\frac{p}{n}}\Big)\,,\label{ref1111}\\
       \explain\textrm{by $\|\theta\|_2\ge \lambda$; see Equation (\ref{snrgtr1})}\\
       &\le \frac{\|\theta\|^2_2}{2} \nn \\
       \explain \textrm{by $\lambda\le \frac{c\|\theta\|_2}{1+(p/n)^{1/4}}$ for small enough $c$}
 \end{align}
 Note that the leading eigenvalue of $\frac{\|\theta\|_2^2\eta\eta^T}{n}$ is
\begin{align*}
    \lambda_1\Big(\frac{\|\theta\|^2_2\eta\eta^T}{n}\Big) = \|\theta\|^2_2,
\end{align*}
with the leading eigenvector being $\frac{\eta}{\sqrt{n}}$.  
 Therefore, by Davis-Kahan's theorem (e.g.,  \cite[Theorem 20]{chi2019nonconvex}), the leading eigenvector of $H(\tilde{Y}^T\tilde{Y})$ denoted by $\tilde{v}$ satisfies 
\begin{align}\label{Z2bound}
    \min\left\{\Big\|\hat{v}-\frac{\eta}{\sqrt{n}}\Big\|_2,\Big\|\hat{v}+\frac{\eta}{\sqrt{n}}\Big\|_2\right\}\lesssim \frac{\|Z_2\|_{op}}{\|\theta\|^2_2}\lesssim   \frac{1}{n} + \frac{\lambda}{\|\theta\|_2}+\frac{\lambda^2}{\|\theta\|_2^2}\sqrt{\frac{p}{n}}\,.
\end{align} 
Recall that $\tilde{\eta}^0=\sign(\tilde{v})$, and therefore 
\begin{align*}
   \frac{\min\{\|\tilde{\eta}^0-\eta\|_1,\|\tilde{\eta}^0+\eta\|_1\}}{n}&\le 2\min\left\{\Big\|\hat{v}-\frac{\eta}{\sqrt{n}}\Big\|_2^2,\Big\|\hat{v}+\frac{\eta}{\sqrt{n}}\Big\|^2_2\right\}\\\explain\textrm{by Lemma \ref{retractbound}}\\
    &\lesssim  \frac{1}{n^2}+ \frac{\lambda^2}{\|\theta\|_2^2}+\frac{\lambda^4}{\|\theta\|_2^4}\frac{p}{n}.\\
    \explain\textrm{by Equation (\ref{Z2bound})}
\end{align*}
By using $r_{n,\lambda}=\frac{\|\theta\|_2^2/\lambda^2}{\sqrt{\|\theta\|_2^2/\lambda^2+p/n}}$, it follows that 
\begin{align}\label{finalinibound}
    \mathbb{P}\bigg(\ell(\tilde{\eta}^0,\eta)\le C\Big[\frac{1}{n^2}+\frac{1}{r_{n,\lambda}^2}\Big]\bigg)\ge 1-4\exp(-2n)
\end{align}
for some absolute constant $C$.

\subsubsection*{Step 2.2. Analyzing the Local Refinement}

We now analyze the iterates of Lloyd's algorithm.  We define the following events
\begin{gather}\label{Aidefine}
    A_i:= \Big\{\Big(\frac{H(\tilde{Y}^T\tilde{Y})_i^T}{n}\eta\Big)\eta_i\ge \frac{C^*\|\theta\|^2_2}{r_{n,\lambda}}\Big\},\quad i =1,...,n.
\end{gather}
where $H(\tilde{Y}^T\tilde{Y})_i^T$ denotes the $i$-th row of $H(\tilde{Y}^T\tilde{Y})$, $C^*$ can be chosen to be large enough. We then define
\begin{align}\label{eb}
    B := \bigg\{\|Z_2\|_{op}\le C_1\left(\!\frac{\|\theta\|^2_2}{n}+ \lambda\|\theta\|_2+ \lambda^2\sqrt{\frac{p}{n}}\right)\bigg\} 
\end{align}
that holds for some absolute constant $C_1$ with probability at least $1-4\exp(-2n)$ --- this is a consequence of the analysis in Step 2.1, see Equation (\ref{ref1111}). Furthermore, we define 
 the event 
\begin{align}\label{ec}
    C:=\bigg\{\frac{1}{n}\sum_{i=1}^n \mathbf{1}({A_i^c}) \le \frac{C_2}{r_{n,\lambda}^{2}}\bigg\},
\end{align} 
Moreover, we define
\begin{align}
    B':= \bigg\{\ell(\tilde{\eta}^0,\eta) \le C_3\Big[\frac{1}{n^2}+\frac{1}{r_{n,\lambda}^2}\Big]\bigg\}
\end{align}
that holds for some absolute constant $C_3$ with probability at least $1-4\exp(-2n)$; see Equation (\ref{finalinibound}). Without loss of generality, we assume $\ell(\tilde{\eta}^0,\eta)=\frac{1}{n}\|\tilde{\eta}^0-\eta\|_1$ and can write
\begin{align} \label{eventbprime}
     B' =\bigg\{\frac{\|\tilde{\eta}^0-\eta\|_1}{n}\le C_3\Big[\frac{1}{n^2}+\frac{1}{r_{n,\lambda}^2}\Big]\bigg\}.
\end{align}
By the above discussion, 
\begin{align}
    \label{probBBprime}
    \mathbb{P}(B) \ge 1-4\exp(-2n),\quad\textrm{and}\quad \mathbb{P}(B')\ge 1-4\exp(2n). 
\end{align}
\paragraph{Iterates Stay Near $\eta$ on $S:=B\cap B'\cap C$.} We first prove
\begin{align}
    \frac{\|\tilde{\eta}^k-\eta\|_1}{n}\mathbf{1}(S) \le \frac{C_4}{r_{n,\lambda}^2} + \frac{C_4}{n^2},\quad \forall k\ge 0\label{induct1}
\end{align}
for some large enough constant $C_4$ by induction. 
 By $B'$ in Equation (\ref{eventbprime}), (\ref{induct1}) is trivial for $k=0$. Now, we suppose
\begin{align}
    \label{hypo}
    \frac{\|\tilde{\eta}^k-\eta\|_1}{n}\mathbf{1}(S) \le \frac{C_4}{r_{n,\lambda}^2} + \frac{C_4}{n^2}
\end{align}
and seek to prove
\begin{align*}
    \frac{\|\tilde{\eta}^{k+1}-\eta\|_1}{n}\mathbf{1}(S) \le \frac{C_4}{r_{n,\lambda}^2} + \frac{C_4}{n^2}\,.
\end{align*}
By $\tilde{\eta}^{k+1}= \sign(H(\tilde{Y}^T\tilde{Y})\tilde{\eta}^{k})$, the $i$-th entry of $\tilde{\eta}^{k+1}$ is given by
\begin{align} 
    (\tilde{\eta}^{k+1})_i = \sign\left(H(\tilde{Y}^T\tilde{Y})_i^T\tilde{\eta}^k\right)=\sign\left(H\Big(\frac{1}{n}\tilde{Y}^T\tilde{Y}\Big)_i^T\tilde{\eta}^k\right).
\end{align}
We now compute $H(\frac{1}{n}\tilde{Y}^T\tilde{Y})_i^T\tilde{\eta}^k$ in the following:  
\begin{align}
  &  H(\frac{1}{n}\tilde{Y}^T\tilde{Y})_i^T\tilde{\eta}^k=  H(\frac{1}{n}\tilde{Y}^T\tilde{Y})_i^T (\tilde{\eta}^k-\eta)+ H(\frac{1}{n}\tilde{Y}^T\tilde{Y})_i^T \eta\nn \\\nn
  & = (Z_2)_i^T (\tilde{\eta}^k-\eta) + \frac{\|\theta\|_2^2}{n}\eta_i\eta^T(\tilde{\eta}^k-\eta) + H(\frac{1}{n}\tilde{Y}^T\tilde{Y})_i^T \eta  \\ \explain\textrm{by $H(\frac{1}{n}\tilde{Y}^T\tilde{Y})=\frac{\|\theta\|_2^2}{n}\eta\eta^T+Z_2$ from Equation (\ref{HYYZ2})}\\\nn 
  & =  (Z_2)_i^T (\tilde{\eta}^k-\eta) - \frac{\eta_i\|\theta\|_2^2}{n}\big(n-\eta^T\tilde{\eta}^k\big)+ H(\frac{1}{n}\tilde{Y}^T\tilde{Y})_i^T \eta\\\explain\textrm{by $\eta^T(\tilde{\eta}^k-\eta) = \eta^T\tilde{\eta}^k-n$}  \\
  & =  (Z_2)_i^T (\tilde{\eta}^k-\eta) - \eta_i\|\theta\|^2_2\frac{\|\tilde{\eta}^k-\eta\|_1}{n}+ H(\frac{1}{n}\tilde{Y}^T\tilde{Y})_i^T \eta.\label{computerowc}
  \\\explain\textrm{by $\|\tilde{\eta}^k-\eta\|_1 = \frac{1}{2}\|\tilde{\eta}^k-\eta\|_2^2  = n-\eta^T\tilde{\eta}^k$} 
\end{align}
Note that the error of $\tilde{\eta}$ can be written as 
\begin{align*}
    \frac{\|\tilde{\eta}^{k+1}-\eta\|_1}{n} = \frac{2}{n}\bigg(\sum_{i:\eta_i=1}\mathbf{1}({H(\tilde{Y}^T\tilde{Y})_i^T \tilde{\eta}^k<0})+\sum_{i:\eta_i=-1}\mathbf{1}({H(\tilde{Y}^T\tilde{Y})_i^T \tilde{\eta}^k>0})\bigg)
\end{align*}
and therefore\begin{align}\label{errordecom}
    \frac{\|\tilde{\eta}^{k+1}-\eta\|_1}{n}\mathbf{1}(S)= \frac{2}{n}\bigg(\sum_{i:\eta_i=1}\mathbf{1}({H(\tilde{Y}^T\tilde{Y})_i^T \tilde{\eta}^k<0})\mathbf{1}(S)+\sum_{i:\eta_i=-1}\mathbf{1}({H(\tilde{Y}^T\tilde{Y})_i^T \tilde{\eta}^k>0})\mathbf{1}(S)\bigg). 
\end{align} 
We separately discuss the cases of $\eta_i=1$ and $\eta_i=-1$:
\begin{itemize}[leftmargin=5ex,topsep=0.25ex]
    \setlength\itemsep{-0.3em}
    \item If $\eta_i=1$, then by (\ref{computerowc}) 
    \begin{align*}
         H(\frac{1}{n}\tilde{Y}^T\tilde{Y})_i^T \tilde{\eta}^k=(Z_2)_i^T (\tilde{\eta}^k-\eta) - \|\theta\|^2_2\frac{\|\tilde{\eta}^k-\eta\|_1}{n}+ H(\frac{1}{n}\tilde{Y}^T\tilde{Y})_i^T \eta
    \end{align*}
    and the event $A_i$ defined in Equation (\ref{Aidefine}) reads 
    \begin{align}\label{currentAi}
        A_i = \bigg\{ \frac{1}{n}H(\tilde{Y}^T\tilde{Y})_i^T\eta\ge \frac{C^*\|\theta\|^2_2}{r_{n,\lambda}}\bigg\}.
    \end{align}
    Therefore, 
    \begin{align*}
        &\mathbf{1}({H(\tilde{Y}^T\tilde{Y})_i^T \tilde{\eta}^k<0})\le \mathbf{1}({H(\frac{1}{n}\tilde{Y}^T\tilde{Y})_i^T \tilde{\eta}^k<0,A_i})+\mathbf{1}({A_i^c})\\
        &\le \mathbf{1}\left((Z_2)_i^T(\tilde{\eta}^k-\eta)\le -\frac{C^*\|\theta\|^2_2}{r_{n,\lambda}}+\frac{\|\theta\|^2_2\|\tilde{\eta}^k-\eta\|_1}{n}\right)+\mathbf{1}(A_i^c)\\
        \explain\textrm{by Equations (\ref{computerowc}) and (\ref{currentAi})} 
    \end{align*}
    which further leads to 
    \begin{align*}
        &\mathbf{1}({H(\tilde{Y}^T\tilde{Y})_i^T \tilde{\eta}^k<0})\mathbf{1}(S)\\&=\mathbf{1}\left((Z_2)_i^T(\tilde{\eta}^k-\eta)\le -\frac{C^*\|\theta\|^2_2}{r_{n,\lambda}}+\frac{\|\theta\|^2_2\|\tilde{\eta}^k-\eta\|_1}{n}\,,~S\right)+\mathbf{1}(A_i^c,S) \\
        &\le \mathbf{1}\left((Z_2)_i^T(\tilde{\eta}^k-\eta)\le -\frac{C^*\|\theta\|^2_2}{2r_{n,\lambda}}+\frac{C_4\|\theta\|^2_2}{n^2},~S\right)+\mathbf{1}(A_i^c,S)\\\explain\textrm{by Equation (\ref{hypo})}\\
        &\le \mathbf{1}\left((Z_2)_i^T(\tilde{\eta}^k-\eta)\le -\frac{C^*\|\theta\|^2_2}{3r_{n,\lambda}},\,S\right)+\mathbf{1}(A_i^c,S)\\
        \explain \textrm{by the assumption $\underbrace{\|\theta\|_2\le \lambda n \Big(1+\frac{p^{1/4}}{n^{3/4}}\Big)}_{\textrm{due to $r_{n,\lambda}\lesssim n$}}$ and large enough $C^*$}
    \end{align*}
   Further using $\mathbf{1}(T\le -a)\le \frac{T^2}{a^2}$ for any $a>0$,  we arrive at 
    \begin{align}\label{etai1case}
        \mathbf{1}({H(\tilde{Y}^T\tilde{Y})_i^T \tilde{\eta}^k<0})\mathbf{1}(S) \le \frac{[3r_{n,\lambda}(Z_2)_i^T(\tilde{\eta}^k-\eta)]^2}{(C^*)^2\|\theta\|_2^4} \mathbf{1}(S) +\mathbf{1}(A_i^c,S)
    \end{align}
    \item If $\eta_i=-1$, then by argument parallel to the last dot point, we obtain
     \begin{align}\label{etai-1case}
        \mathbf{1}({H(\tilde{Y}^T\tilde{Y})_i^T \tilde{\eta}^k>0})\mathbf{1}(S) \le \frac{[3r_{n,\lambda}(Z_2)_i^T(\tilde{\eta}^k-\eta)]^2}{(C^*)^2\|\theta\|_2^4} \mathbf{1}(S) +\mathbf{1}(A_i^c,S)
    \end{align}
\end{itemize}

The above disucssions then yield  
\begin{align}\nn 
    &\frac{\|\tilde{\eta}^{k+1}-\eta\|_1}{n}\mathbf{1}(S) \le \frac{2}{n}\sum_{i=1}^n \left( \frac{[3r_{n,\lambda}(Z_2)_i^T(\tilde{\eta}^k-\eta)]^2}{(C^*)^2\|\theta\|^4_2} \mathbf{1}(S) +\mathbf{1}(A_i^c,S)\right)\\ \explain \textrm{by substituting (\ref{etai1case}) and (\ref{etai-1case}) into (\ref{errordecom})}\\ \nn 
    &\le \frac{18}{n} \frac{r_{n,\lambda}^2\sum_{i=1}^n[(Z_2)_i^T(\tilde{\eta}^k-\eta)]^2}{(C^*)^2\|\theta\|^4_2}\mathbf{1}(S)  + \frac{2}{n}\sum_{i=1}^n \mathbf{1}(A_i^c,S)\\\explain\textrm{by simple algebra}\\
    &\le \frac{18}{n} \frac{r_{n,\lambda}^2\|Z_2\|_{op}^2\|\tilde{\eta}^k-\eta\|_2^2}{(C^*)^2\|\theta\|^4_2}\mathbf{1}(S) + \frac{2}{n}\sum_{i=1}^n\mathbf{1}(A_i^c)\mathbf{1}(S)\label{weneed}\\\explain\textrm{by $\sum_{i=1}^n[(Z_2)_i^T(\tilde{\eta}^k-\eta)]^2=\|Z_2(\tilde{\eta}^k-\eta)\|_2^2\le \|Z_2\|_{op}^2\|\tilde{\eta}^k-\eta\|_2^2$}\\ \nn 
    & \le\frac{\|\tilde{\eta}^k-\eta\|_1\mathbf{1}(S)}{n}O\bigg(\frac{r_{n,\lambda}^2 }{(C^*)^2\|\theta\|^4_2}\Big(\frac{\|\theta\|^4_2}{n^2}+\lambda^2\|\theta\|^2_2+ {\lambda^4} \frac{p}{n}\Big)\bigg)  + O(r_{n,\lambda}^{-2})\\
    \explain\textrm{by $\|\tilde{\eta}^k-\eta\|_1=\frac{1}{2}\|\tilde{\eta}^k-\eta\|_2^2$ and the events in Equations (\ref{eb}) and (\ref{ec})}\\ 
    &{\le}\Big(\frac{1}{r_{n,\lambda}^2}+\frac{1}{n^2}\Big) \underbrace{O\left(\frac{r_{n,\lambda}^2 }{(C^*)^2}\left(\frac{1}{n^2}+\frac{\lambda^2}{\|\theta\|^2_2}+ \frac{\lambda^4}{\|\theta\|_2^4} \frac{p}{n}\right)\right)}_{:=F}  + O(r_{n,\lambda}^{-2}) \nn \\
    \explain\textrm{by the hypothesis in (\ref{hypo})}
 \end{align}
 We now examine the scaling of $F$:  
since
$r_{n,\lambda}^2\big(\frac{\lambda^2}{\|\theta\|^2_2}+\frac{\lambda^4}{\|\theta\|^4_2}\frac{p}{n}\big)=\Theta(1)$, we have 
$ F \lesssim \frac{1}{(C^*)^2}\big(\frac{r_{n,\lambda}^2}{n^2}+\Theta(1)\big)$; moreover, $\|\theta\|_2\le \lambda(n+(np)^{1/4})$ implies $r_{n,\lambda}=O(n)$ and therefore $F=O(\frac{1}{(C^*)^2})$; therefore, we can set $C^*$ large enough to render a small enough $F$, yielding 
\begin{align*}
    \frac{\|\tilde{\eta}^{k+1}-\eta\|_1}{n}\mathbf{1}(S)  \le  \frac{C_4}{r_{n,\lambda}^2} + \frac{C_4}{n^2} 
\end{align*}
as long as $C_4$ is suitably large. 
  The induction is complete and we have established \begin{align}
    \frac{\|\tilde{\eta}^k-\eta\|_1}{n}\mathbf{1}(S) \le \frac{C_4}{r_{n,\lambda}^2} + \frac{C_4}{n^2},\quad k\ge 0.\label{induct11}
\end{align}

\paragraph{Yielding Sharp Bound.}  
  We shall proceed with an intermediate result in the above induction argument, which also holds for all $k\ge 0$: by Equation (\ref{weneed}), we have
\begin{align*}
    \frac{\|\tilde{\eta}^{k+1}-\eta\|_1}{n}\mathbf{1}(S)\le \frac{36r_{n,\lambda}^2\|Z_2\|_{op}^2}{(C^*)^2\|\theta\|^4_2}\cdot\frac{\|\tilde{\eta}^k-\eta\|_1}{n}\mathbf{1}(S)+\frac{2}{n}\sum_{i=1}^n\mathbf{1}(A_i^c)\mathbf{1}(S)\,,\quad\forall k\ge 0.
\end{align*} 
As before, we substitute the bound on $\|Z_2\|_{op}$ from Equation (\ref{eb}) and take $C^*$ large enough to obtain
\begin{align*}
     \frac{\|\tilde{\eta}^{k+1}-\eta\|_1}{n}\mathbf{1}(S)\le \frac{1}{4}\cdot  \frac{\|\tilde{\eta}^{k}-\eta\|_1}{n}\mathbf{1}(S)+\frac{2}{n}\sum_{i=1}^n\mathbf{1}(A_i^c)\mathbf{1}(S)\,,\quad\forall k\ge 0.
\end{align*}
We then iterate this inequality, along with $\frac{\|\tilde{\eta}_k-\eta\|_1}{n}\le 2$, to yield 
\begin{align*}
    \frac{\|\tilde{\eta}^{k}-\eta\|_1}{n}\mathbf{1}(S) \le \frac{2}{4^k}+\frac{C_5}{n}\sum_{i=1}^n\mathbf{1}(A_i^c)\mathbf{1}(S),\quad\forall k\ge 0.
\end{align*}
Thus, it holds that 
\begin{align}\label{deterlargek}
    \frac{\|\tilde{\eta}^{k}-\eta\|_1}{n}\mathbf{1}(S) \le \frac{1}{n^2}+\frac{C_5}{n}\sum_{i=1}^n\mathbf{1}(A_i^c)\mathbf{1}(S),\qquad \forall k\ge \lceil 3\log_4 n\rceil. 
\end{align}  We now take expectation, 
\begin{align}\nn
    \mathbb{E} \frac{\|\tilde{\eta}^k-\eta\|_1}{n}  &\le \mathbb{E}\bigg(\frac{\|\tilde{\eta}^k-\eta\|_1}{n}\mathbf{1}(S)\bigg)+\mathbb{E}\bigg(\frac{\|\tilde{\eta}^k-\eta\|_1}{n}\Big[\mathbf{1}(B^c)+\mathbf{1}\big((B')^c\big)+\mathbf{1}(C^c)\Big]\bigg)\\ \explain\textrm{by $S=B\cap B'\cap C$ and hence $S^c = B^c\cup (B')^c\cup C^c$}\\ \nn 
    &\le \frac{1}{n^2}+\mathbb{E}\Big(\frac{C_5}{n}\sum_{i=1}^n\mathbf{1}(A_i^c)\Big) + 2\big(\mathbb{P}(B^c)+\mathbb{P}\big((B')^c\big)+\mathbb{P}(C^c)\big)
    \\\explain\textrm{by Equation (\ref{deterlargek}) and $\frac{\|\tilde{\eta}^k-\eta\|_1}{n}\le 2$}\\
    &\le \frac{1}{n^2} + C_5\mathbb{P}(A_i^c)+16\exp(-2n)+2\mathbb{P}(C^c) ,\quad \forall k\ge \lceil 3\log_4 n\rceil.\label{expcbound}\\
    \explain\textrm{by $\mathbb{P}(B^c)+\mathbb{P}((B')^c)\le 8\exp(-2n)$ from (\ref{probBBprime})}
\end{align}
Recall from Equation (\ref{ec}) that 
$C^c= \big\{\frac{1}{n}\sum_{i=1}^n\mathbf{1}(A_i^c)>\frac{C_2}{r_{n,\lambda}^2}\big\}$. Thus,
by Markov's inequality we have 
\begin{align*}
    &\frac{C_2\mathbb{P}(C^c)}{r_{n,\lambda}^2}\le \mathbb{E}\Big(\frac{1}{n}\sum_{i=1}^n\mathbf{1}(A_i^c)\Big) = \mathbb{P}(A_i^c) 
    \\\Longrightarrow~&\mathbb{P}(C^c)= O(r_{n,\lambda}^{2}\mathbb{P}(A_i^c)).
\end{align*}
Substituting this into (\ref{expcbound}), along with $r_{n,\lambda}\gtrsim 1$, yields
\begin{align}\label{finalexpbound}
    \mathbb{E}\frac{\|\tilde{\eta}^k-\eta\|_1}{n} \le \frac{2}{n^2} + C'r_{n,\lambda}^{2}\mathbb{P}(A_i^c),\qquad \forall k\ge \lceil 3\log_4 n\rceil
\end{align}
where the sufficiently large $n$ ensures $16\exp(-2n)\le\frac{1}{n^2}$. Recall from Equation (\ref{Aidefine}) that
\begin{align}\label{evAic}
    A_i^c = \Big\{\Big(\frac{H(\tilde{Y}^T\tilde{Y})_i^T}{n}\eta\Big)\eta_i < \frac{C^*\|\theta\|_2^2}{r_{n,\lambda}}\Big\}.
\end{align}
By noticing  
$$H(\tilde{Y}^T\tilde{Y})_i^T = \big(\langle \tilde{X}_i,\tilde{X}_1 \rangle,\cdots, \langle \tilde{X}_i,\tilde{X}_{i-1}\rangle,0,\langle \tilde{X}_i,\tilde{X}_{i+1} \rangle, \cdots,\langle \tilde{X}_i,\tilde{X}_n \rangle\big),$$
we have that 
\begin{align*}
    \Big(\frac{H(\tilde{Y}^T\tilde{Y})_i^T}{n}\eta\Big)\eta_i = \frac{1}{n} \sum_{j\ne i}\eta_i\eta_j \langle\tilde{X}_i,\tilde{X}_j\rangle = \Big\langle \eta_i\tilde{X}_i,\frac{1}{n}\sum_{j\ne i}\eta_j\tilde{X}_j\Big\rangle.
\end{align*}
By $\varepsilon_i\stackrel{d}{=}-\varepsilon_i$\footnote{We use $A\stackrel{d}{=} B$ to denote that $A$ and $B$ have the same distribution.} and $\tilde{X}_i = \eta_i\theta + u_i$ where $u_i=Q_{2\lambda}(\eta_i\theta + \varepsilon_i+\tau_i)-\eta_i\theta$, we have
\begin{align*}
    \eta_i\tilde{X}_i =  \theta+ \eta_iu_i &= \theta +(Q_{2\lambda}(\theta+\eta_i\varepsilon_i+\eta_i\tau_i)-\theta)\\
    &\stackrel{d}{=} \theta+ (Q_{2\lambda}(\theta+\varepsilon_i+\tau_i)-\theta).
\end{align*}
As a result, we can let $\delta_1,\cdots,\delta_n$ be i.i.d. and follow the distribution $Q_{2\lambda}(\theta+\varepsilon_i+\tau_i)-\theta$, then we have  
\begin{align}
    \Big(\frac{H(\tilde{Y}^T\tilde{Y})_i^T}{n}\eta\Big)\eta_i\stackrel{{\rm d}}{=}\big\langle \theta + \delta_1,\frac{1}{n}\sum_{j=2}^n (\theta+ \delta_j)\big\rangle\label{samedisss}
\end{align}
and thus
\begin{align}\nn
    \mathbb{P}(A_i^c) &= \mathbb{P} \bigg(\Big\langle \theta+\delta_1,\frac{1}{n}\sum_{j=2}^n (\theta+\delta_j)\Big\rangle < \frac{C^*\|\theta\|_2^2}{r_{n,\lambda}}\bigg)\\\explain\textrm{by substituting (\ref{samedisss}) into (\ref{evAic})}\\
    &\le \mathbb{P} \bigg(\Big\langle \theta+\delta_1,\theta+\frac{1}{n-1}\sum_{j=2}^n \delta_j\Big\rangle < \frac{2C^*\|\theta\|_2^2}{r_{n,\lambda}}\bigg):= P^*,\label{definePstar} \\
    \explain\textrm{by $\frac{1}{n-1}<\frac{2}{n}$}
\end{align}
where $2C^*$ is some large enough absolute constant. Substituting this into (\ref{finalexpbound}) yields \[\textrm{$\mathbb{E}\Big(\frac{1}{n}\|\tilde{\eta}^k-\eta\|_1\Big)\le \frac{2}{n^2}+ C'r_{n,\lambda}^2 P^*$ ~~for any~ $k\ge \lceil 3\log_4 n\rceil$},\] 
and   Markov's inequality further yields  
\begin{align*}
    \mathbb{P}\Big(\frac{\|\tilde{\eta}^k-\eta\|_1}{n} \ge \frac{2\sqrt{r_{n,\lambda}}}{n^2} + C'r_{n,\lambda}^{5/2}P^*\Big)\le r_{n,\lambda}^{-1/2}.
\end{align*}
Combining with $r_{n,\lambda}=O(n)$, we reach
\begin{align}\label{tildeetabound}
    \mathbb{P}\bigg(\frac{\|\tilde{\eta}^k-\eta\|_1}{n}  \le \frac{C}{n^{3/2}}+C'r_{n,\lambda}^{5/2}P^*\bigg)\ge 1-r_{n,\lambda}^{-1/2}.
\end{align}
We have now completed the analysis of Algorithm \ref{alg:2lloyd} with  $\tilde{Y}$.

\subsubsection*{Step 3. Reduction to One-Bit Quantization} 
All that remains is to  reduce to the one-bit quantization regime by transitioning from   $\tilde{\eta}^k$ (obtained by running Algorithm \ref{alg:2lloyd} with $\dot{Y}$ being replaced by $\tilde{Y}$) to $\hat{\eta}^k$ (which is obtained by running Algorithm \ref{alg:2lloyd}). To this end, we only need to show $\lambda\dot{Y}=\tilde{Y}$, namely 
\begin{align}\label{tildeequaldot}
     \lambda\sign(\eta_i\theta+ \varepsilon_i + \tau_i) = \lambda\dot{X}_i = \tilde{X}_i = Q_{2\lambda}(\eta_i\theta+\varepsilon_i+\tau_i),\quad \forall i\in[n].
\end{align} 
The reason is that, on the event of (\ref{tildeequaldot}),  $\hat{\eta}^k=\tilde{\eta}^k$ holds for all $k\ge0$, and therefore the bound for $\tilde{\eta}^k$ in Equation (\ref{tildeetabound}) transfers to $\hat{\eta}^k$ of interest.
By the observation 
\begin{align*}
    Q_{2\lambda}(a) = \lambda\sign(a)\,,\quad\forall |a|\le 2\lambda,
\end{align*}  
it suffices to show $ \|\eta_i\theta + \varepsilon_i + \tau_i\|_\infty \le 2\lambda$ for all $i\in[n]$. Using triangle inequality and combining with $\|\tau_i\|_\infty\le \lambda$, it suffices to show 
\begin{align}\label{Wmaxbound}
    \|W\|_{\max}:=\sup_{i\in [n]}\|\varepsilon_i \|_\infty\le \lambda - \|\theta\|_\infty,\qquad \forall i\in [n].
\end{align}
By the tail bound of the entries of $W$ (Lemma \ref{lem:sgtail}) and a union bound, it follows that
\begin{align*}
    \mathbb{P}\Big(\|W\|_{\max}\ge \sigma\sqrt{2(1+\nu)\log(np)}\Big)\le \frac{2}{(np)^\nu}.
\end{align*}
As we assume $\lambda- \|\theta\|_\infty\ge \sigma\sqrt{2(1+\nu)\log(np)}$, the desired result follows.

\subsubsection*{Step 4. Sharp Bound on $P^*$} 
  In light of
    \begin{align*}
        \bigg\langle\theta+\delta_1,\theta+\frac{1}{n-1}\sum_{j=2}^n\delta_j\bigg\rangle=\|\theta\|_2^2+ \bigg\langle\theta,\frac{1}{n-1}\sum_{j=2}^n \delta_j\bigg\rangle + \bigg\langle\delta_1,\theta+\frac{1}{n-1}\sum_{j=2}^n\delta_j\bigg\rangle
    \end{align*}
    and the definition of $P^*$ in Equation (\ref{definePstar}), 
    we have
    \begin{align}\label{Pstarformulate}
        P^*= \mathbb{P}\bigg(\underbrace{\bigg\langle\theta,\frac{1}{n-1}\sum_{j=2}^n \delta_j\bigg\rangle + \bigg\langle\delta_1,\theta+\frac{1}{n-1}\sum_{j=2}^n\delta_j\bigg\rangle <\Big(-1+\frac{C_1}{r_{n,\lambda}}\Big)\|\theta\|_2^2}_{\textrm{denoted by an event $\Gamma$}}\bigg). 
    \end{align}
    We first show that $\langle \theta, \frac{1}{n-1}\sum_{j=2}^n\delta_j\rangle$ is a higher order term. Since $\delta_j\stackrel{{\rm d}}{=} Q_{2\lambda}(\theta+\varepsilon_j+\tau_j)-\theta$ has independent, zero mean entries with $O(\lambda)$ sub-Gaussian norms (e.g., by argument parallel to Equation (\ref{uisgbound})), we can use Lemma \ref{ver261} to obtain 
    \begin{align}\label{adapthm31}
        \bigg\|\frac{1}{n-1}\sum_{j=2}^n\delta_j\bigg\|_{\psi_2}\lesssim \frac{\lambda}{\sqrt{n}}~\Longrightarrow ~ \bigg\|\bigg\langle \theta, \frac{1}{n-1}\sum_{j=2}^n\delta_j\bigg\rangle\bigg\|_{\psi_2}\lesssim \frac{\lambda\|\theta\|_2}{\sqrt{n}}.
    \end{align}
    Therefore, let \begin{align}
        \hat{B}:=\bigg\{\Big|\Big\langle\theta,\frac{1}{n-1}\sum_{j=2}^n\delta_j\Big\rangle\Big|<C\lambda\|\theta\|_2\sqrt{\frac{\log n}{n}}\bigg\}\textrm{ ~with large enough $C$,}\label{evehatB}
    \end{align}   then by (\ref{adapthm31}) and a sub-Gaussian tail bound, 
    \begin{align}\label{adapthm32}
        \mathbb{P}(\hat{B}^c)=\mathbb{P}\bigg(\bigg|\Big\langle\theta,\frac{1}{n-1}\sum_{j=2}^n\delta_j\Big\rangle\bigg|\ge C\lambda\|\theta\|_2\sqrt{\frac{\log n}{n}} \bigg) \le\frac{1}{n^2}.
    \end{align}
    Moreover, we let 
    \begin{align}\label{evehatBone}
        \hat{B}_1:= \bigg\{\Big\|\frac{1}{n-1}\sum_{j=2}^n\delta_j\Big\|_2^2 \le \big(1+O\Big(\frac{1}{n}\Big)\big)\frac{p\lambda^2}{n}+O\Big(\frac{1}{\sqrt{n}}\Big)\|\theta\|_2^2\bigg\}, 
    \end{align}
    then   Lemma \ref{lem:boundsumvnorm} gives
    \begin{align}\label{B1cboundbound}
        \mathbb{P}(\hat{B}_1^c)\le 2\exp(-10r_{n,\lambda}^2). 
    \end{align} 
    On the event $\hat{B}\cap \hat{B}_1$, along with $\|\theta\|_2\gtrsim\lambda$, we have
    \begin{align}
        \bigg\|\theta+\frac{1}{n-1}\sum_{j=2}^n\delta_j\bigg\|_2^2 &=\|\theta\|_2^2+\bigg\|\frac{1}{n-1}\sum_{j=2}^n\delta_j\bigg\|_2^2 + 2\bigg\langle \theta,\frac{1}{n-1}\sum_{j=2}^n\delta_j\bigg\rangle  \nn\\\explain\textrm{by simple algebra} \\
        &\le \|\theta\|_2^2 + \Big(1+O(\frac{1}{n})\Big)\frac{p\lambda^2}{n} + O(\frac{1}{\sqrt{n}})\|\theta\|_2^2 + 2C\lambda\|\theta\|_2\sqrt{\frac{\log n}{n}}\nn\\\explain\textrm{by $\hat{B}$ in (\ref{evehatB}) and $\hat{B}_1$ in (\ref{evehatBone})}\\
        & = \left(1+O\Big( \sqrt{\frac{\log n}{n}}\Big)\right)\|\theta\|_2^2 + \left(1+O\Big(\frac{1}{\sqrt{n}}\Big)\right)\frac{p\lambda^2}{n}. \label{sharpnormbb}\\\explain\textrm{by $\lambda\le \|\theta\|_2$; see Equation (\ref{snrgtr1})}
    \end{align}
Moreover, on the event of $\hat{B}\cap \hat{B}_1$, we have 
    \begin{align}
         &\text{The event $\Gamma$ in Equation (\ref{Pstarformulate})}\nn
         \\\nn 
         \Longrightarrow~& 
          \bigg\langle\delta_1,\theta+\frac{1}{n-1}\sum_{j=2}^n\delta_j\bigg\rangle <\Big(-1+\frac{C_1}{r_{n,\lambda}}\Big)\|\theta\|_2^2+ \bigg|\Big\langle\theta,\frac{1}{n-1}\sum_{j=2}^n \delta_j\Big\rangle\bigg|\\ \nn 
        \Longrightarrow~&  \bigg\langle\delta_1,\theta+\frac{1}{n-1}\sum_{j=2}^n\delta_j\bigg\rangle <\Big(-1+\frac{C_1'}{r_{n,\lambda}}\Big)\|\theta\|_2^2\\\explain\textrm{by $\hat{B}$ in Equation (\ref{evehatB})} \\
          \Longrightarrow~& \bigg\langle \delta_1, \frac{\theta+(n-1)^{-1}\sum_{j=2}^n\delta_j}{\|\theta+(n-1)^{-1}\sum_{j=2}^n\delta_j\|_2}\bigg\rangle <  
          -\frac{\Big(1-O(r_{n,\lambda}^{-1})-O((\frac{\log n}{n})^{1/4})\Big)\|\theta\|_2^2}{\sqrt{\|\theta\|_2^2+\frac{p\lambda^2}{n}}}. \label{eventimply}\\\explain\textrm{by Equation (\ref{sharpnormbb}) that holds under $\hat{B}\cap\hat{B}_1$}
    \end{align}
Therefore, by letting $\hat{\xi}:=\frac{\theta+(n-1)^{-1}\sum_{j=2}^n\delta_j}{\|\theta+(n-1)^{-1}\sum_{j=2}^n\delta_j\|_2}=(\hat{\xi}_1,\cdots,\hat{\xi}_p)^T$, 
    \begin{align} 
        &P^*=\mathbb{P}(\Gamma)\le \mathbb{P}(\hat{B}^c)+\mathbb{P}(\hat{B}^c_1)  + \mathbb{P}(\Gamma\cap \hat{B}\cap \hat{B}_1) \nn \\
        &\le \frac{1}{n^2}+2\exp(-10r_{n,\lambda}^2)+\underbrace{\mathbb{P}\bigg(\langle \delta_1,\hat{\xi}\rangle<-\Big(1-O\Big(r_{n,\lambda}^{-1}+(\frac{\log n}{n})^{1/4}\Big)\Big)\frac{\|\theta\|_2^2}{\sqrt{\|\theta\|_2^2+\frac{p\lambda^2}{n}}}\bigg)}_{:=P^{**}}\label{Pstarstardefine}\\\explain\textrm{by Equations (\ref{adapthm32}), (\ref{B1cboundbound}) and (\ref{eventimply})}
    \end{align}
    To establish sharp bound on $P^{**}$, 
    we need a reduction to the one-bit case. Recall that \[\delta_1\stackrel{{\rm d}}{=}Q_{2\lambda}(\theta+\varepsilon_1+\tau_1)-\theta,\] we shall introduce a surrogate of $\delta_1$ as 
    \begin{align*}
        \tilde{\delta}_1:=\lambda\sign(\theta+\varepsilon_1+\tau_1) - \mathbb{E}\big[\lambda\sign(\theta+\varepsilon_1+\tau_1)\big]=(\tilde{\delta}_{11},\cdots,\tilde{\delta}_{1p})^T.
    \end{align*}
    By Lemma \ref{lem:Pstarstar}, for some $\epsilon_{n,\lambda}=O(r_{n,\lambda}^{-1}+(\frac{\log n}{n})^{1/4})$,
    \begin{align}\nn
        P^{**}&\le \mathbb{P}\bigg(\langle \tilde{\delta}_1,\hat{\xi}\rangle<-\frac{(1-\epsilon_{n,\lambda})\|\theta\|_2^2}{\sqrt{\|\theta\|_2^2+\frac{p\lambda^2}{n}}}\bigg)+\frac{1}{n^{1+\nu}p^\nu}\\&= \mathbb{P}\bigg(\sum_{j=1}^p\frac{\tilde{\delta}_{1j}\hat{\xi}_j}{\lambda} <-(1-\epsilon_{n,\lambda})r_{n,\lambda}\bigg)+\frac{1}{n^{1+\nu}p^\nu}\,. \label{P2starbound}\\\explain\textrm{recall $r_{n,\lambda}=\frac{\|\theta\|_2^2}{\lambda^2}\big/\sqrt{\frac{\|\theta\|_2^2}{\lambda^2}+\frac{p}{n}}$}
    \end{align}
    Since $\{\tilde{\delta}_{1j}:j\in[p]\}$ are independent, zero-mean, and $\|\lambda\sign(\theta+\varepsilon_1+\tau_1)\|_\infty\le \lambda$, we shall use Hoeffding's inequality (see Lemma \ref{hoeffdinglem})   to obtain 
    \begin{align}
        \label{hoeffdings}\mathbb{P}\bigg(\sum_{j=1}^p\frac{\tilde{\delta}_{1j}\hat{\xi}_j}{\lambda} <-(1-\epsilon_{n,\lambda})r_{n,\lambda}\bigg)\le 2\exp\bigg(-\frac{(1-\epsilon_{n,\lambda})^2r_{n,\lambda}^2}{2}\bigg). 
    \end{align}

\subsubsection*{Step 5. Putting Pieces Together} 
 We substitute  (\ref{hoeffdings}) and (\ref{P2starbound}) into (\ref{Pstarstardefine}) to obtain
 \[P^*\le  2\exp\bigg(-\frac{(1-\epsilon_{n,\lambda})^2r_{n,\lambda}^2}{2}\bigg)+\frac{1}{n^2}+2\exp(-10r_{n,\lambda}^2) + \frac{1}{n^{1+\nu}p^\nu}\]
 Substituting this into (\ref{tildeetabound}) it follows that   
     \begin{align*}
         \ell(\hat{\eta}^t ,\eta)& \le Cr_{n,\lambda}^{5/2}\exp\bigg(-\frac{(1-\epsilon_{n,\lambda})r_{n,\lambda}^2}{2}\bigg)+ O\bigg(\frac{1}{n^{3/2}}+ \frac{r_{n,\lambda}^{5/2}}{n^2} + \frac{r_{n,\lambda}^{5/2}}{n^{1+\nu}p^\nu}\bigg)\\\explain \textrm{by $\exp(-10r_{n,\lambda}^2)\le \exp(-\frac{1}{2}r_{n,\lambda^2})$} \\ 
         &\le \frac{1}{2}\exp\bigg(-\Big(1-\epsilon_{n,\lambda}-\frac{2\log(2Cr_{n,\lambda}^{5/2})}{r_{n,\lambda}^2}\Big)\frac{r_{n,\lambda}^2}{2}\bigg) + \frac{1}{5n}\\\explain\textrm{by $n\gtrsim 1,~r_{n,\lambda}\lesssim \min\{(np)^{2\nu/5},n^{2/5}\}$}\\
         & =\frac{1}{2} \exp\bigg(-\frac{(1-\epsilon'_{n,\lambda})r_{n,\lambda}^2}{2}\bigg) + \frac{1}{5n}
     \end{align*}
     for some $\epsilon_{n,\lambda},\epsilon_{n,\lambda}'=O(\frac{1}{r_{n,\lambda}}+(\frac{\log n}{n})^{1/4})$. If $\frac{1}{2}\exp(-\frac{(1-\epsilon_{n,\lambda}')r_{n,\lambda}^2}{2})\le \frac{1}{5n}$, then we have $\ell(\hat{\eta}^t,\eta) \le \frac{2}{5n}$ and thus necessarily $\ell(\hat{\eta}^t,\eta) =0$. If $\frac{1}{2}\exp(-\frac{(1-\epsilon_{n,\lambda}')r_{n,\lambda}^2}{2})>\frac{1}{5n}$, then we have
     \begin{align*}
         \ell(\hat{\eta}^t,\eta)\le \exp \bigg(-\frac{(1-\epsilon'_{n,\lambda})r_{n,\lambda}^2}{2}\bigg),
     \end{align*}
     as desired. 
 \end{proof}
 \section{Proof of Corollary \ref{cor:exact} (Exact Recovery)}\label{app:proveexact1}
 \begin{proof}
   Note that the separation condition in Equation (\ref{sepaexact}) implies 
   \[\|\theta\|_2^2 \gtrsim \lambda^2 \Big(\log n + \sqrt{\frac{p\log n}{n}}\Big).\]
  In light of $r_{n,\lambda}=\frac{\|\theta\|_2^2/\lambda^2}{\sqrt{\|\theta\|_2^2/\lambda^2+p/n}}\asymp \min\{\frac{\|\theta\|_2}{\lambda},\frac{\|\theta\|^2}{\lambda^2}\sqrt{\frac{n}{p}}\}$, we then have 
   \[\frac{1}{r_{n,\lambda}}\lesssim \max\{\frac{\lambda}{\|\theta\|_2},\frac{\lambda^2}{\|\theta\|_2^2}\sqrt{\frac{p}{n}}\}=O(\frac{1}{\sqrt{\log n}}).\]
   Under $n\gtrsim 1$, Theorem \ref{maintheorem} guarantees that \[\ell(\hat{\eta}^t,\eta)\le \exp(-\frac{(1-\epsilon_0)}{2}r_{n,\lambda}^2),\quad\textrm{for some $\epsilon_0=O(\frac{1}{\sqrt{\log n}})$}\] holds with probability at least $1-O(\frac{1}{\sqrt{\log n}})-2(np)^{-\nu}$. Note that $\ell(\hat{\eta}^t,\eta)<n^{-1}$ implies $\ell(\hat{\eta}^t,\eta)=0$. Therefore, to ensure $\hat{\eta}^t=\pm\eta$ with the promised probability, it suffices to have \begin{align}\nn
       &\exp\Big(-\frac{(1-\epsilon_0)}{2}r_{n,\lambda}^2\Big)<\frac{1}{n}\\
       \nn\iff~&\exp\bigg(-\frac{1-\epsilon_0}{2}\frac{(\|\theta\|_2^2/\lambda^2)^2}{\frac{\|\theta\|_2^2}{\lambda^2}+\frac{p}{n}}\bigg) < \frac{1}{n} \\
       \nn\iff~& \frac{1-\epsilon_0}{2}\frac{(\|\theta\|_2^2/\lambda^2)^2}{\frac{\|\theta\|_2^2}{\lambda^2}+\frac{p}{n}} > \log n \\
       \nn\iff~& (1-\epsilon_0)\Big(\frac{\|\theta\|_2^2}{\lambda^2}\Big)^2 -2\log n \frac{\|\theta\|_2^2}{\lambda^2} -\frac{2p\log n}{n}>0\\
       \label{quadraticcalculate}\iff~& \|\theta\|_2^2 > \frac{1}{1-\epsilon_0}\lambda^2 \bigg(1+\sqrt{1+\frac{(1-\epsilon_0)2p }{n\log n}}\bigg)\log n
   \end{align}
   This condition can be implied by (\ref{sepaexact}) by  noticing $\frac{1}{1-\epsilon_0}=1+\frac{\epsilon_0}{1-\epsilon_0}=1+O(\epsilon_0)=1+O(\frac{1}{\sqrt{\log n}})$ and  $1-\epsilon_0\le 1$. 
\end{proof}
 \section{Proof of Theorem \ref{thm:spifree} (Partial Recovery without Spikiness Condition)}\label{app:proofthm2}
 We start with a remark that summarizes the additional technicalities beyond the proof of Theorem \ref{maintheorem}. 
 \begin{rem}
      The result is proved by revisiting  the proof of Theorem \ref{maintheorem} and similarly relating $\lambda\dot{X}_i=\lambda\sign(\eta_iR\theta + R\varepsilon_i+
\tau_i)$ to $\tilde{X}_i:=Q_{2\lambda}(\eta_iR\theta+R\varepsilon_i+\tau_i)$. Yet  the overall noise   $u_i=Q_{2\lambda}(\eta_iR\theta+R\varepsilon_i+\tau_i)-\eta_iR\theta$ has possibly correlated entries, making some arguments no longer valid. Our general remedy is to further decompose $u_i$ into $u_i = \big[Q_{2\lambda}(\eta_iR\theta+R\varepsilon_i+\tau_i)-(\eta_iR\theta+R\varepsilon_i)\big]+R\varepsilon_i := u_{i1}+R\varepsilon_i$ 
and then separately treat $u_{i1}$ and $R\varepsilon_i$: note that $u_{i1}$ has independent zero-mean entries (when  conditioning on $\varepsilon_i$) and $R\varepsilon_i$ is simply a rotation of    $\varepsilon_i$ with independent zero-mean entries. Following this idea, most arguments can be readily extended by triangle inequality. \label{rem:proofthm2}
 \end{rem}
 \begin{proof}
 By letting $\dot{Y}=[\dot{X}_1,...,\dot{X}_n]$, $Y=[X_1,...,X_n]$, $D=[\tau_1,...,\tau_n]$, $W=[\varepsilon_1,...,\varepsilon_n]$, 
we can write\begin{align*}
    \dot{Y} = \sign(RY+D) = \sign(R\theta\eta^T+RW+D).
\end{align*}
  The differences from Theorem \ref{maintheorem} are that the center becomes $\theta_R:=R\theta$, and that the noise matrix becomes $W_R:=RW=[R\varepsilon_1,\cdots,R\varepsilon_n]$, which has independent columns but may not have independent rows. By Lemma \ref{Rthetamu}, we assume that $\mu(\theta_R)\le \sqrt{\frac{3\log p}{p}}$ holds with the promised probability. This proof will follow that of Theorem \ref{maintheorem} with necessary adaptations.

\paragraph{Step 1. Checking Steps 1--3 in Appendix \ref{app:proofthm1}.} We now re-iterate the proof of Theorem \ref{maintheorem} to reach the following: with probability at least $1-r_{n,\lambda}^{-1/2}-2(np)^{-\nu}-3p^{-1/4}$, we have that 
\begin{align} \label{Pstarbounderr}
    \ell(\hat{\eta}^t,\eta) \le \frac{C'}{n^{3/2}}+C'' r_{n,\lambda}^{5/2}P^*, 
\end{align}
where $P^*$ is defined as  
\begin{align}\label{evePstar}
    P^*:= \mathbb{P}\left(\left\langle \theta+\delta_1,\theta+\frac{1}{n-1}\sum_{j=2}^n \delta_j\right\rangle < \frac{C_1\|\theta\|_2^2}{r_{n,\lambda}}\right)
\end{align}
where $\delta_1,\cdots,\delta_n$ are i.i.d. random vectors in $\mathbb{R}^p$ having the same distribution as $Q_{2\lambda}(\theta_R+R\varepsilon_i + \tau_i)-\theta_R$. In fact, we only need to modify the concentration bounds, as detailed in the following: 
\begin{itemize}[leftmargin=5ex,topsep=0.25ex]
    \setlength\itemsep{-0.3em}
 \item Note that we  now have $\tilde{Y}=Q_{2\lambda}(R\theta\eta^T+RW+D)=R\theta\eta^T+U$ and thus $U=Q_{2\lambda}(R\theta\eta^T+RW+D)-R\theta\eta^T$.  
    \item  Since $U$ has independent columns of $O(\lambda)$ sub-Gaussian norms, the argument in Equations (\ref{Uthetal2})--(\ref{etaubound}) for bounding $\|U^T\theta\|_2$ remains valid.

    \item Since $U$ has possibly correlated rows,  Equation (\ref{utuboundeq}) via Lemma \ref{utubound} is no longer valid, and we should use Lemma \ref{lem:genutu} to achieve the same bound, by the decomposition   
 \begin{align}\label{decomposeU}
     U = \underbrace{Q_{2\lambda}(R\theta\eta^T+RW+D)-(R\theta\eta^T+RW)}_{:=U_1}+RW:=U_1+RW 
 \end{align} 
 where, given $R\in \mathbb{O}(p)$,  $W$ has i.i.d. entries of $O(\sigma)$ sub-Gaussian norms, and $U_1|W$ also has i.i.d. entries of $O(\lambda)$ sub-Gaussian norms. 
 \item It is evident that (\ref{Wmaxbound}) remains true for $\|RW\|_{\max}$. 
\end{itemize}

\paragraph{Step 2. Checking Steps 4--5 in Appendix \ref{app:proofthm1}.} 
The additional challenge is that $\delta_1,...,\delta_n$ may not have independent entries. Analogously to (\ref{decomposeU}), we decompose $\{\delta_j\}_{j=1}^n$ as 
\begin{align}\label{decomdeltaj}
    \delta_j = \underbrace{Q_{2\lambda}(\theta_R+R\varepsilon_j+\tau_j)-(\theta_R+R\varepsilon_j)}_{:=\delta'_j}+R\varepsilon_i
\end{align}
and separately treat $\delta'_j$ and $R\varepsilon_i$. It turns out that the adaptation of Equations (\ref{adapthm31})--(\ref{adapthm32}) is simple, and we still have \[\hat{B}:=\bigg\{\Big|\langle \theta,\frac{1}{n-1}\sum_{j=2}^n\delta_j\rangle\Big|\le C\lambda\|\theta\|_2\sqrt{\frac{\log n}{n}}\bigg\}\] holds with probability at least $1-\frac{1}{n^2}$. By Lemma \ref{lem:sumnormlem}, we also have that \[\hat{B}_1:=\Big\{\Big\|\frac{1}{n-1}\sum_{j=2}^n\delta_j\Big\|_2^2\le (1+O(\frac{1}{n}))\frac{p\lambda^2}{n}+O(\frac{1}{\sqrt{n}})\|\theta\|_2^2\Big\}\] holds with probability at least $1-6\exp(-10r_{n,\lambda}^2)$. On the event $\hat{B}\cap \hat{B}_1$, following the derivations of Equations (\ref{sharpnormbb}) and (\ref{eventimply}), we have 
$$\bigg\|\theta+\frac{1}{n-1}\sum_{j=2}^n\delta_j\bigg\|_2^2\le \bigg(1+O\Big(\sqrt{\frac{\log n}{n}}\Big)\bigg)\|\theta\|_2^2 + \left(1+O\left(\frac{1}{\sqrt{n}}\right)\right)\frac{p\lambda^2}{n},$$
and thus the event defining $P^*$ in Equation (\ref{evePstar}) implies \[
   \lambda^{-1} \langle \delta_1,\hat{\xi}\rangle \le -(1-O(r_{n,\lambda}^{-1})-O\Big((\frac{\log n}{n})^{1/4})\Big)r_{n,\lambda}\] for a vector $\hat{\xi}\in\mathbb{S}^{p-1}$ independent of $\delta_1$; therefore,  
   \begin{align*}
       P^* \le \mathbb{P}\left(\frac{\langle\delta_1,\hat{\xi}\rangle}{\lambda}\le -\Big(1-O(r_{n,\lambda}^{-1})-O((\frac{\log n}{n})^{1/4})\Big)r_{n,\lambda}\right),
   \end{align*}
   and further using $\delta_1=\delta'_1+R\varepsilon_1$ in Equation (\ref{decomdeltaj}) yields
   \begin{align}\nn
       P^* &\le \mathbb{P}\left(\frac{\langle \delta_1',\hat{\xi}\rangle}{\lambda}\le -\left(1-O(r_{n,\lambda}^{-1})-O((\frac{\log n}{n})^{1/4})-\frac{1}{\log^{1/4}(np)}\right)r_{n,\lambda}\right)\\\nn
       &+ \mathbb{P}\left(\frac{\langle\varepsilon_1,R^T\hat{\xi}\rangle}{\lambda}\le -\frac{r_{n,\lambda}}{\log^{1/4}(np)}\right)+\frac{1}{n^2}+6\exp(-10r_{n,\lambda}^2)\\:&=P^{**}+P^*_2+\frac{1}{n^2}+6\exp(-10r_{n,\lambda}^2). \label{Psssdefine}
   \end{align}
   By the sub-Gaussian tail bound, $ P_2^*\le \exp(-\sqrt{\log(np)}\cdot r_{n,\lambda}^2)$, which is negligible. To obtain sharp bound on $P^{**}$, we need a reduction to the one-bit case. As a surrogate of \[\delta_1'=Q_{2\lambda}(\theta_R+R\varepsilon_1+\tau_1)-(\theta_R+R\varepsilon_1),\] we introduce $\tilde{\delta}_1'=\lambda\sign(\theta_R+R\varepsilon_1+\tau_1)-\mathbb{E}_{\tau_1}[\lambda\sign(\theta_R+R\varepsilon_1+\tau_1)]$. We invoke Lemma \ref{lem:deltatotilehaar} to reach \[P^{**}\le \mathbb{P}\Big(\frac{\langle\tilde{\delta}_1',\hat{\xi}\rangle}{\lambda}\le -(1-\epsilon_{n,\lambda})r_{n,\lambda}\Big)+\frac{1}{n^{1+\nu}p^\nu}.\] Furthermore, by conditioning on $\varepsilon_1$ and using the randomness of $\tau_1$ (so that the entries of $\tilde{\delta}_1'$ are independent, zero-mean), we can invoke Hoeffding's inequality as in Equation (\ref{hoeffdings}) to reach 
\[\mathbb{P}\bigg(\frac{\langle\tilde{\delta}_1',\hat{\xi}\rangle}{\lambda}\le -(1-\epsilon_{n,\lambda})r_{n,\lambda}\bigg)\le 2\exp\Big(-\frac{(1-\epsilon_{n,\lambda})^2r_{n,\lambda}^2}{2}\Big)=2\exp\Big(-\frac{1}{2}(1-\epsilon_{n,\lambda}')r_{n,\lambda}^2\Big)\] for some $\epsilon_{n,\lambda}'=O(r_{n,\lambda}^{-1}+(\frac{\log n}{n})^{1/4}+\frac{1}{\log^{1/4}(np)})$. Substituting this and $P_2^*\le \exp(-\sqrt{\log(np)}\cdot r_{n,\lambda}^2)$ into (\ref{Psssdefine}) yields  $$P^*\le 3\exp\Big(-\frac{1}{2}(1-\epsilon_{n,\lambda}')r_{n,\lambda}^2\Big)+\frac{1}{n^2}+6\exp(-10r_{n,\lambda}^2)+\frac{1}{n^{1+\nu}p^\nu}.$$ Further combining with (\ref{Pstarbounderr}), along with $r_{n,\lambda}\lesssim \min\{n^{2/5},(np)^{2\nu/5}\}$   and some simple reasoning used in the proof of Theorem \ref{maintheorem}, yields the claim. 
\end{proof}
\section{Proof of Corollary \ref{coro:haar} (Exact Recovery without Spikiness Condition)}\label{app:provecoro2} 
\begin{proof}
     In view of (\ref{minimallambda}) and (\ref{separation111}), we have $\|\theta\|_2^2\gtrsim \lambda^2(\log n+ \sqrt{\frac{p\log n}{n}})$. Combining with $r_{n,\lambda}=\frac{\|\theta\|_2^2/\lambda^2}{\sqrt{\|\theta\|_2^2/\lambda^2+p/n}}\asymp \min\{\frac{\|\theta\|_2}{\lambda},\frac{\|\theta\|^2}{\lambda^2}\sqrt{\frac{n}{p}}\}$, we then have 
   \[\frac{1}{r_{n,\lambda}}\lesssim \max\{\frac{\lambda}{\|\theta\|_2},\frac{\lambda^2}{\|\theta\|_2^2}\sqrt{\frac{p}{n}}\}=O(\frac{1}{\sqrt{\log n}}).\] 
   As a result, Theorem \ref{thm:spifree} yields that for some $\epsilon_{n,\lambda}= O(\log^{-1/4}n)$ and for any $t\ge \lceil 3\log_4n\rceil$, it holds with the promised probability that 
   \[\ell(\hat{\eta}^t,\eta) \le \exp\bigg(-\frac{(1-\epsilon_{n,\lambda})r_{n,\lambda}^2}{2}\bigg).\]
   This leads to the desired $\hat{\eta}^t=\pm \eta$ if 
   \begin{align*}
       &\exp\bigg(-\frac{(1-\epsilon_{n,\lambda})r_{n,\lambda}^2}{2}\bigg)<\frac{1}{n}\\
       \iff~& \|\theta\|_2^2 > \frac{1}{1-\epsilon_{n,\lambda}}\lambda^2 \bigg(1+\sqrt{1+\frac{(1-\epsilon_0)2p}{n\log n}}\bigg)\log n. \\
       \explain\textrm{by calculations parallel to Equation (\ref{quadraticcalculate})}
   \end{align*}
   This can be implied by 
   \[\|\theta\|_2^2  \ge (1+\epsilon) \lambda^2\bigg(1+\sqrt{1+\frac{2p}{n\log n}}\bigg)\log n,\quad\textrm{for some $\epsilon=O(\log^{-1/4}n)$},\]
   which is exactly what we assume in Equations (\ref{minimallambda}) and (\ref{separation111}).

  It remains to show that the condition (\ref{separation111}) can be implied by the more explicit condition (\ref{oursepara}) under $\min\{n,p,r_{n,\lambda}\}\ge C_1(\nu)$ and $p\ge C_2(\nu)\log n\log\log n$ for some sufficiently large constants $C_1(\nu), C_2(\nu)$ depending on $\nu$ only. For some sufficiently large $C_0(\nu)$, we use the inequality 
  \[(A_1+A_2)^2 =A_1^2+A_2^2+2A_1A_2\le (1+C_0(\nu))A_1^2+(1+\frac{1}{C_0(\nu)})A_2^2\]
  to achieve 
  \begin{align*}
        &\bigg(\sqrt{\frac{3\log p}{p}}\|\theta\|_2+\sigma\sqrt{2(1+\nu)\log (np)}\bigg)^2\\
        &\le \frac{3[1+C_0(\nu)]\log p\|\theta\|_2^2}{p} + \Big[1+\frac{1}{C_0(\nu)}\Big](1+\nu) 2\sigma^2\log(np)\\
        &\le \frac{6C_0(\nu)\log p\|\theta\|_2^2}{p} +  (1+1.5\nu) 2\sigma^2\log(np).\\\explain\textrm{by choosing $C_0(\nu)$ large enough to ensure this}
  \end{align*}
In turn a sufficient condition for Equation (\ref{separation111}) is 
\[\|\theta\|_2^2\ge (1+O(\log^{-1/4}n))\bigg(\frac{6C_0(\nu)\log p\|\theta\|_2^2}{p} +  (1+1.5\nu) 2\sigma^2\log(np)\bigg)\bigg(1+\sqrt{1+\frac{2p}{n\log n}}\bigg)\log n,\]
and by rearranging, this is implied by 
\begin{align}\nn
    &\bigg(1-\tilde{C}_0(\nu)\frac{\log n\log p}{p}\big[1+\sqrt{\frac{p}{n\log n}}\big]\bigg)\|\theta\|_2^2\\
    \label{cccc}&\qquad\ge \big(1+O(\log^{-1/4}n)\big)(1+1.5\nu)2\sigma^2\log(np)\bigg(1+\sqrt{1+\frac{2p}{n\log n}}\bigg)\log n
\end{align}
for some   constant $\tilde{C}_0(\nu)$ depending on $\nu$ only. Note that (\ref{cccc})  can be written as 
\[\|\theta\|_2^2\ge \underbrace{\frac{\big(1+O(\log^{-1/4}n)\big)(1+1.5\nu)}{\big(1-\tilde{C}_0(\nu)\frac{\log n\log p}{p}\big[1+\sqrt{\frac{p}{n\log n}}\big]\big)}}_{F_1}2\sigma^2\bigg(1+\sqrt{1+\frac{2p}{n\log n}}\bigg)\log (np)\log n\]
All that remains is to show that the leading factor $F_1$ is bounded by $1+2\nu$. This is ensured by the following observations: (i) $n\ge C_1(\nu)$ with large enough $C_1(\nu)$ ensures that $O(\log ^{-1/4}n)\le c_1(\nu)$ for some small enough $c_1(\nu)$; (ii) $p\ge C_2(\nu)\log n\log\log n$  with large enough $C_2(\nu)$ implies that $\tilde{C}_0(\nu)\frac{\log n\log p}{p}\le  c_2(\nu)$ for some small enough $c_2(\nu);$ (iii) $n,p\ge C_1(\nu)$ for large enough $C_1(\nu)$ ensures that
\[\frac{\tilde{C}_0(\nu)\log n\log p}{p} \sqrt{\frac{p}{n\log n}} =\tilde{C}_0(\nu)\sqrt{\frac{\log^2p\log n}{np}}\le c_3(\nu)\quad\textrm{for some small enough $c_3(\nu).$}\]
The proof is complete. 
\end{proof}
\section{Proof of Theorem \ref{thmlower} (Minimax Lower Bound)} \label{app:prooflower}
\begin{proof}
We fix $\theta = \frac{\|\theta\|_2}{\sqrt{p}}\mathbf{1}_p$ and place an independent Rademacher prior on the label $\eta$, meaning that we let $\eta_1,\cdots,\eta_n$ be i.i.d. uniformly distributed over $\{-1,1\}$. Let $\mathbb{E}_{\pi}$ be the expectation on the prior, then we have
\begin{align*}
    \inf_{\hat{\eta}}\sup_{\theta,\eta}\, \mathbb{E} \big[\ell(\hat{\eta},\eta)\big] \ge \inf_{\hat{\eta}}\,\mathbb{E}_\pi \mathbb{E} \big[\ell(\hat{\eta},\eta)\big].
\end{align*}
We let $T\subset [n]$ be of size $\lfloor n/2\rfloor + 1$, then we can handle the sign flipping associated with $\ell(\hat{\eta},\eta)$ by argument in \cite{gao2018community}, reaching
\begin{align*}
    \inf_{\hat{\eta}}\,\mathbb{E}_\pi \mathbb{E} \big[\ell(\hat{\eta},\eta)\big] &\ge \frac{c}{|T^c|}\sum_{i\in T^c}\inf_{\tilde{\eta}^i}\mathbb{E}_\pi \mathbb{E}|\tilde{\eta}^i(\dot{Y},(\eta_j)_{j\in T})-\eta_i|\\
    &\ge  \frac{c}{|T^c|}\sum_{i\in T^c}\inf_{\tilde{\eta}^i}\mathbb{E}_\pi \mathbb{E}|\tilde{\eta}^i(\dot{Y},(\eta_j)_{j\ne i})-\eta_i| \\
    \explain\textrm{since $(\eta_j)_{j\ne i}$ contains more information than $(\eta_j)_{j\in T}$}
\end{align*}
for some constant $c$. Moreover, since the pairs $(\dot{Y}_j,\eta_j)$ are independently distributed, we have 
\begin{align}\label{eq11}
     \inf_{\hat{\eta}}\,\mathbb{E}_\pi \mathbb{E} \big[\ell(\hat{\eta},\eta)\big] \ge c \cdot \inf_{\bar{\eta}_j}\,\mathbb{E}_\pi \mathbb{E}|\bar{\eta}_j(\dot{X}_j)-\eta_j|,
\end{align}
where $\eta_i$ is a Rademacher variable, and $\bar{\eta}_i$ is a measurable function of $\dot{X}_i$ for estimating $\eta_i$. We now let \begin{gather*}
    P_1 = \mathbb{P}\Big(\sign\Big(\frac{\|\theta\|_2}{\sqrt{p}}+\varepsilon_{ij}+\tau_{ij}\Big) = 1\Big),\\
    P_{-1} = \mathbb{P}\Big(\sign\Big(\frac{\|\theta\|_2}{\sqrt{p}}+\varepsilon_{ij}+\tau_{ij}\Big)=-1\Big), 
\end{gather*}
which satisfy $P_1>\frac{1}{2}>P_{-1}$ by the symmetry of $\varepsilon_{ij}$ and $\tau_{ij}$. Then, it is easy to write the likelihood of $\dot{X}_i$ under $\eta_i = 1$ and $\eta_i = -1$ as
\begin{gather*}
    f_1(\dot{Y}_j) = P_1^{|\{i\in [p]:\dot{Y}_{ij}=1\}|}P_{-1}^{p-|\{i\in [p]:\dot{Y}_{ij}=1\}|},\\ f_{-1}(\dot{Y}_j) = P_1^{p-|\{i\in [p]:\dot{Y}_{ij}=1\}|}P_{-1}^{|\{i\in [p]:\dot{Y}_{ij}=1\}|}.
\end{gather*}
Therefore, by Bayesian decision rule, the selector 
\begin{align*}
    \eta_k^* = \mathbf{1}(|\{i\in[p]:\dot{Y}_{ij}=1\}|\ge \frac{p}{2}) - \mathbf{1}(|\{i\in[p]:\dot{Y}_{ij}=1\}|<\frac{p}{2})=\sign(\dot{Y}_j^T\mathbf{1})
\end{align*}
attains the minimum of the right-hand side of (\ref{eq11}). Then, it is not hard to see that
\begin{align}\label{anticon}
    \mathbb{E}_\pi \mathbb{E}|\bar{\eta}_j(\dot{X}_j)-\eta_j| = 2 \mathbb{P}\Big(\text{Binomial}(p,P_{-1})>\frac{p}{2}\Big) \ge \frac{1}{\sqrt{2p}}\exp\Big(-pD_{KL}\Big(\frac{1}{2}\|P_{-1}\Big)\Big)
\end{align}
where the last inequality is due to the anti-concentration bound for binomial variable (cf. Lemma \ref{lem:anti}).

We now lower bound $P_{-1}$ as follows:
\begin{align}\nn
    & P_{-1} = \mathbb{P}\Big(\frac{\|\theta\|_2}{\sqrt{p}}+\varepsilon_{ij}+\tau_{ij}<0\Big)\\\nn
    & \ge \mathbb{E}\left(\mathbf{1}\Big(\Big|\frac{\|\theta\|_2}{\sqrt{p}}+\varepsilon_{ij}\Big|\le\lambda \Big)\mathbf{1}\Big(\tau_{ij}<-\varepsilon_{ij}-\frac{\|\theta\|_2}{\sqrt{p}}\Big)\right)\\
    \nn&  = \mathbb{E}\left(\mathbf{1}\Big(\Big|\frac{\|\theta\|_2}{\sqrt{p}}+\varepsilon_{ij}\Big|\le\lambda \Big)\cdot\frac{\lambda-\varepsilon_{ij}-p^{-1/2}\|\theta\|_2}{2\lambda}\right) \\\explain\textrm{by the randomness of $\tau_{ij}\sim{\rm Unif}[-\lambda,\lambda]$}\\
    \nn& =  \mathbb{E}\left(\frac{\lambda-\varepsilon_{ij}-p^{-1/2}\|\theta\|_2}{2\lambda}\right) - \mathbb{E}\left(\frac{\lambda-\varepsilon_{ij}-p^{-1/2}\|\theta\|_2}{2\lambda}\mathbf{1}\Big(\Big|\frac{\|\theta\|_2}{\sqrt{p}}+\varepsilon_{ij}\Big|>\lambda\Big)\right)\\
    & \ge  \frac{1}{2}- \frac{\|\theta\|_2}{2\lambda\sqrt{p}} - E_1-E_2, \label{E1E21}
\end{align}
where we let 
\begin{gather*}
    E_1 : = \mathbb{E}\left(\left|\frac{\lambda-p^{-1/2}\|\theta\|_2}{2\lambda}\right|\mathbf{1}\Big(\Big|\frac{\|\theta\|_2}{\sqrt{p}}+\varepsilon_{ij}\Big|>\lambda\Big)\!\right),\\
    E_2: = \mathbb{E}\left(\frac{|\varepsilon_{ij}|}{2\lambda}\mathbf{1}\Big(\Big|\frac{\|\theta\|_2}{\sqrt{p}}+\varepsilon_{ij}\Big|>\lambda\Big)\!\right).
\end{gather*}
By standard Gaussian tail bound, 
\begin{align*}
    E_1 \le \frac{1}{2}\mathbb{P}\Big(|\varepsilon_{ij}|>\lambda - \frac{\|\theta\|_2}{\sqrt{p}}\Big) \le \exp\Big(-\frac{(\lambda-p^{-1/2}\|\theta\|_2)^2}{2\sigma^2}\Big).
\end{align*}
Next, we have
\begin{align*}
    E_2 & \le \frac{1}{2\lambda}\mathbb{E}\big(|\varepsilon_{ij}|\mathbf{1}(|\varepsilon_{ij}|>\lambda-\|\theta\|_2/\sqrt{p})\big)\\
    & = \frac{\sigma}{2\lambda}\int_{\sigma^{-1}(\lambda-\frac{\|\theta\|_2}{\sqrt{p}})}^\infty \sqrt{\frac{2}{\pi}} s \exp(-\frac{s^2}{2})\,\text{d}s\\
    & = \frac{\sigma}{\lambda\sqrt{2\pi}}\exp\Big(-\frac{(\lambda-p^{-1/2}\|\theta\|_2)^2}{2\sigma^2}\Big).
\end{align*}
Substituting these bounds into (\ref{E1E21}) yields
\begin{align*}
    P_{-1} \ge \frac{1}{2}- \left(\frac{\|\theta\|_2}{2\lambda\sqrt{p}}+\Big(1+\frac{\sigma}{\lambda\sqrt{2\pi}}\Big)\exp\Big(-\frac{(\lambda-p^{-1/2}\|\theta\|_2)^2}{2\sigma^2}\Big)\right): = \frac{1}{2}-\xi_0,
\end{align*}
where we let 
\begin{align*}
    \xi_0 :&= \frac{\|\theta\|_2}{2\lambda\sqrt{p}}+\Big(1+\frac{\sigma}{\lambda\sqrt{2\pi}}\Big)\exp\Big(-\frac{(\lambda-p^{-1/2}\|\theta\|_2)^2}{2\sigma^2}\Big) ,
\end{align*}
and under $\lambda\gtrsim\sigma + p^{-1/2}\|\theta\|_2$ and $\log(\frac{2\lambda\sqrt{p}}{\|\theta\|_2})\le \frac{\lambda^2}{8\sigma^2}$, we have
\begin{align}
    \xi_0 & \le \frac{\|\theta\|_2}{2\lambda\sqrt{p}} + 2\exp\Big(-\frac{\lambda^2}{4\sigma^2}\Big) = \frac{\|\theta\|_2}{2\lambda\sqrt{p}}\Big(1+2\exp\Big(-\frac{\lambda^2}{8\sigma^2}\Big)\Big). \label{xi_0asym}
\end{align}

Note that $\xi_0$ is small enough under the assumption $\lambda\ge C(\frac{\|\theta\|_2}{\sqrt{p}}+\sigma)$. Continuing from (\ref{anticon}), 
\begin{align*}
     \mathbb{E}_\pi \mathbb{E}\,|\bar{\eta}_j(\dot{X}_j)-\eta_j|  &\ge \frac{1}{\sqrt{2p}}\exp\Big(-pD_{KL}\Big(\frac{1}{2}\big\|\frac{1}{2}-\xi_0\Big)\Big)\\
    &= \frac{1}{\sqrt{2p}}\exp\Big(\frac{p}{2}\log(1-4\xi_0^2)\Big)\\
    &\ge \frac{1}{\sqrt{2p}}\exp\big(-2(1+o(1))p\xi_0^2\big),
\end{align*}
where in the last inequality we use $\log(1-4\xi_0^2)\ge -4(1+o(1))\xi_0^2$ for some $o(1)\to 0$ as $\xi_0\to 0$, which is ensured by $\frac{\lambda}{p^{-1/2}\|\theta\|_2+\sigma}\to \infty$. Further using $\xi_0\le (1+o(1))\frac{\|\theta\|_2}{2\lambda\sqrt{p}}$ from (\ref{xi_0asym}), we arrive at the desired claim. 
\end{proof}

\section{Technical Lemmas} \label{app:lem}
\begin{lem}[Anti-concentration of Binomial variable (e.g., \cite{ash2012information})] \label{lem:anti}Let $X\sim {\rm Binomial}(N,q)$, then for any $k=0,1,\cdots,N$ we have $\mathbb{P}\big(X=k\big)\ge \frac{1}{\sqrt{2N}}\exp\big(-N\cdot D_{\rm KL}(\frac{k}{N}\big\|q)\big)$ 
where $D_{\rm KL}(q\|q'):=q\log(\frac{q}{q'})+(1-q)\log(\frac{1-q}{1-q'})$. 
\end{lem}  
\begin{lem}[Sub-Gaussian tail bound] \label{lem:sgtail}
Under Assumption \ref{subgaussian}, for any $v\in \mathbb{R}
^p$ with $\|v\|_2=1$, we have 
\[\max\bigg\{\mathbb{P}(v^T\varepsilon_i\ge t),\mathbb{P}(v^T\varepsilon_i\le-t)\bigg\}\le \exp\Big(-\frac{\sigma^2 t^2}{2}\Big),\quad \forall t\ge 0.\]
\end{lem}
\begin{proof}
    We only bound $\mathbb{P}(v^T\varepsilon_i\ge t)$ since $\mathbb{P}(v^T\varepsilon_i\le -t)=\mathbb{P}((-v)^T\varepsilon_i\ge t)$. For $t=0$ the bound is trivial. We suppose $t>0$. Then for any $s>0$, 
    \begin{align*}
        \mathbb{P}(v^T\varepsilon_i\ge t) & = \mathbb{P}\Big(\exp(sv^T\varepsilon_i)\ge \exp(st)\Big) \\
        &\le \exp(-st)\mathbb{E}\Big[\exp(sv^T\varepsilon_i)\Big]\\
        \explain\textrm{by Markov's inequality}\\
        & = \exp(-st)\mathbb{E}\bigg[\prod_{j=1}^p\exp(sv_j\varepsilon_{ij})\bigg]\\
        & = \exp(-st)\prod_{j=1}^p \mathbb{E}\big[\exp(sv_j\varepsilon_{ij})\big]\\
        \explain\textrm{by independence of entries of $\varepsilon_i$}\\
        &\le \exp(-st)\prod_{j=1}^p\exp\Big(\frac{\sigma^2 s^2v_j^2}{2}\Big) = \exp\Big(-st+\frac{\sigma^2s^2}{2}\Big).
    \end{align*}
    Setting $s=\frac{t}{\sigma^2}$ yields the claim. 
\end{proof}
\begin{lem}\cite[Proposition 2.6.1]{vershynin2018high} \label{ver261} 
Let $X_1,\cdots,X_N$ be independent, zero-mean variables such that $\|X_i\|_{\psi_2}<\infty$. Then for some universal constant $C$,
\[\bigg\|\sum_{i=1}^NX_i\bigg\|_{\psi_2}^2\le C\sum_{i=1}^N\|X_i\|_{\psi_2}^2\]
\end{lem}
 \begin{lem}\label{psi1psi2}
 \cite[Lemma 2.7.7]{vershynin2018high} Define the sub-exponential   norm (or $\psi_1$-norm) of $X$ as $\|X\|_{\psi_1}=\inf\{K>0:\mathbb{E}\exp(|X|/K)\le 2\}$.  Let $X,Y$ be sub-Gaussian, then $XY$ is sub-exponential: $\|XY\|_{\psi_1}\leq \|X\|_{\psi_2}\|Y\|_{\psi_2}$. 
 \end{lem}
 \begin{lem}\label{centering}
  \cite[Exercise 2.7.10]{vershynin2018high} For some absolute constant $C$, $\|X-\mathbb{E}X\|_{\psi_1}\leq C\|X\|_{\psi_1}$.
 \end{lem}
 \begin{lem}
     \label{hoeffdinglem}\cite[Theorem 2.2.6]{vershynin2018high}
     Let $X_1,...,X_N$ be independent random variables. Assume that $X_i\in[m_i,M_i]$ for every $i$. Then, for any $t>0$ we have
     \[\mathbb{P}\bigg(\sum_{i=1}^N(X_i-\mathbb{E}X_i)\ge t\bigg)\le\exp\bigg(-\frac{2t^2}{\sum_{i=1}^N(M_i-m_i)^2}\bigg)\]
 \end{lem}
  \begin{lem}\label{bern} \cite[Theorem 2.8.1]{vershynin2018high} Let $X_1,...,X_N$ be independent, zero-mean, sub-exponential random variables. Then for every $t\geq 0$, for some absolute constant $c$ we have \begin{equation}
        \nonumber\mathbb{P}\left(\Big|\sum_{i=1}^NX_i\Big|\geq t\right)\leq 2\exp\left(-c\min\Big\{\frac{t^2}{\sum_{i=1}^N \|X_i\|_{\psi_1}^2},\frac{t}{\max_{1\leq i\leq N}\|X_i\|_{\psi_1}}\Big\}\right)
    \end{equation}
 \end{lem}
\begin{lem}\cite[Section 2.1]{abdalla2026robust}
    \label{lem:expect}
    Let $a\in\mathbb{R}$ and $\tau\sim{\rm Unif}[-\lambda,\lambda]$, we have 
    \[\mathbb{E}[\lambda \sign(a+\tau)]=T_{\lambda}(a) :=\begin{cases}
    ~~~~a\,,\quad&\textrm{if}~|a|\le \lambda\\
    \lambda\sign(a)\,,\quad& \textrm{if}~|a|>\lambda 
\end{cases}.\]
\end{lem}
 
\begin{lem}
    \label{Hbound2}Let $H(Z)=Z-\diag(Z)$. For any $Z\in \mathbb{R}^{n\times n}$ we have 
    $\|H(Z)\|_{op}\le 2\|Z\|_{op}$. 
\end{lem}
\begin{proof}
    By triangle inequality, $$\|H(Z)\|_{op}=\|Z-\diag(Z)\|_{op}\le\|Z\|_{op}+\|\diag(Z)\|_{op}\le 2\|Z\|_{op}.$$
\end{proof}
\begin{lem}\cite[Lemma 2]{ndaoud2022sharp}
    \label{lemrow} 
    For a random matrix $W$ with independent, zero-mean columns, we have
    $$\|H(W^TW)\|_{op}\le 2\|W^TW-\mathbb{E}(W^TW)\|_{op}$$
\end{lem} 
\begin{lem}\cite[Lemma 5]{ndaoud2022sharp}\label{retractbound}
    For any $x\in\{-1,1\}^n$ and $y\in\mathbb{R}^n$, we have
    \begin{align*}
        \frac{\|\sign(y)-x\|_1}{n}\le 2\Big\|y-\frac{x}{\sqrt{n}}\Big\|_2^2.
    \end{align*}
\end{lem} 
\begin{lem}
\label{lem:quannoi}
Let $Q_{2\lambda}(a)=2\lambda(\lfloor \frac{a}{2\lambda}\rfloor + \frac{1}{2})$ for $a\in\mathbb{R}$ and $\tau\sim {\rm Unif}[-\lambda,\lambda]$. Then for any $a\in \mathbb{R}$, we have 
\begin{gather*}
    \mathbb{E}\big[Q_{2\lambda}(a+\tau)\big] = a, \quad |Q_{2\lambda}(a+\tau)-a|\le 2\lambda\quad\text{and} \quad \|Q_{2\lambda}(a+\tau)-a\|_{\psi_2}=O(\lambda)
\end{gather*}
\end{lem}
\begin{proof}
    $\mathbb{E}\big[Q_{2\lambda}(a+\tau)\big] = a$ can be found in \cite{gray1993dithered}. For the second result, the definition of $Q_{2\lambda}(a)$ gives $\sup_{a\in\mathbb{R}}|Q_{2\lambda}(a)-a|\le \lambda$, hence 
    \[|Q_{2\lambda}(a+\tau)-a|\le |Q_{2\lambda}(a+\tau)-(a+\tau)|+|\tau|\le 2\lambda.\]
    Since $Q_{2\lambda}(a+\tau)-a$ is zero-mean and bounded by $2\lambda$, $\|Q_{2\lambda}(a+\tau)-a\|_{\psi_2}=O(\lambda)$ immediately follows (e.g., \cite[Section 2]{vershynin2018high}). 
\end{proof}
\begin{lem}\label{lem:sgvbound}
    Let $x\in \mathbb{R}^n$ be a random vector with independent entries with sub-Gaussian norms bounded by $A$. Then for some absolute constant $C$, we have
    \begin{align*}
       \mathbb{P}\big(\|x\|_2\le C\cdot  A\sqrt{n}\big)\ge 1-\exp(-2n).
    \end{align*}
\end{lem}
\begin{proof}
    Let $x_i$ be the $i$-th entry of $x$, then $\|x\|_2^2 = \sum_{j=1}^n x_j^2$, where $x_1^2,...,x_n^2$ are sub-exponential variables with sub-exponent norm bounded by (see Lemma \ref{psi1psi2})
    \begin{align*}
        \|x_i^2\|_{\psi_1}\le\|x_i\|_{\psi_2}^2 \le A^2.
    \end{align*}
    By centering (c.f., Lemma \ref{centering}), we have $\|x_i^2-\mathbb{E}x_i^2\|_{\psi_1}=O(A^2)$.
    Therefore, Bernstein's inequality (cf. Lemma \ref{bern}) gives
    \begin{align*}
        \mathbb{P}\Big(\big|\|x\|_2^2-\mathbb{E}\|x\|_2^2\big|\ge t\Big) \le 2\exp\left(-c \min \left\{\frac{t^2}{nA^4},\frac{t}{A^2}\right\}\right).
    \end{align*}
    Setting $t=C'nA^2$ with large enough $C'$ yields 
    \begin{align*}
        \mathbb{P}\Big(\|x\|_2^2 \ge C'nA^2 + \mathbb{E}\|x\|_2^2\Big)\le 2\exp(-2n).
    \end{align*}
    Moreover, by the sub-Gaussianity of $x_i$ we have $\mathbb{E}x_i^2 =O(A^2)$, and therefore $\mathbb{E}\|x\|_2^2 = O(nA^2)$, which completes the proof. 
\end{proof}
\begin{lem}
    \label{utubound} 
    Suppose that $U\in \mathbb{R}^{p\times n}$ have independent rows with  sub-Gaussian norms bounded by $A$, then for some absolute constant $C$ we have
    \begin{align*}
        \mathbb{P}\bigg(\|U^TU-\mathbb{E}(U^TU)\|_{op}\ge  CA^2n\max\big\{1,\sqrt{\frac{p}{n}}\big\}\bigg) \le 2\exp(-2n). 
    \end{align*}
\end{lem}
\begin{proof}
    To control $\|U^TU-\mathbb{E}(U^TU)\|_{op} = \sup_{v\in \mathbb{S}^{n-1}}\big[\|Uv\|_2^2- \mathbb{E}\|Uv\|_2^2\big]$, we start with a discretization (see, e.g.,  \cite[Exercise 4.4.3]{vershynin2018high}): let $N_{1/3}$ be a $1/3$-net of $\mathbb{S}^{n-1}$ of cardinality bounded by $9^n$, we have
    \begin{align}\label{discrete1}
        \|U^TU-\mathbb{E}(U^TU)\|_{op} \le 3\sup_{v\in N_{1/3}}\,\big[\|Uv\|_2^2- \mathbb{E}\|Uv\|_2^2\big].
    \end{align}
    For any $v\in N_{1/3}$, we note that $Uv=(w_1,...,w_p)^T\in \mathbb{R}^p$ has independent entries with sub-Gaussian norms bounded by $A$, and therefore $\|w_i^2\|_{\psi_1}\le \|w_i\|_{\psi_2}^2=O(A^2)$ by Lemma \ref{psi1psi2}. We now invoke Bernstein's inequality (see Lemma \ref{bern}) to arrive at
    \begin{align*}
        \mathbb{P}\Big(\Big|\|Uv\|_2^2-\mathbb{E}\|Uv\|_2^2\Big|\ge t\Big)&= \mathbb{P}\Big(\Big|\sum_{i=1}^p [w_i^2-\mathbb{E}w_i^2]\Big|\ge t\Big)\\
        & \le 2\exp \Big(-c \min\Big\{\frac{t^2}{pA^4},\frac{t}{A^2}\Big\}\Big),\qquad\forall t\ge 0.
    \end{align*}
    Taking a union bound, along with Equation (\ref{discrete1}), yields
    \begin{align*}
        \mathbb{P}\Big(\|U^TU-\mathbb{E}(U^TU)\|_{op}\ge t\Big)\le 2\exp\Big(n \log 9 - c'\min\Big\{\frac{t^2}{pA^4},\frac{t}{A^2}\Big\}\Big),\quad \forall t\ge 0.
    \end{align*}
    Setting $t = C'A^2n\max\{1,\sqrt{\frac{p}{n}}\}$ with large enough $C'$, we obtain that 
    $$\|U^TU-\mathbb{E}(U^TU)\|_{op}\le C'A^2n\max\{1,\sqrt{\frac{p}{n}}\}$$
    with probability at least $1-2\exp(-2n)$, as claimed.  
\end{proof} 
\begin{lem} \label{lem:boundsumvnorm}
    In the proof of Theorem \ref{maintheorem} in Appendix \ref{app:proofthm1}, we have 
    \begin{align*}
        \mathbb{P}\bigg(\Big\|\frac{1}{n-1}\sum_{j=2}^n\delta_j\Big\|_2^2< \frac{p\lambda^2}{n}\Big(1+\frac{C}{n}\Big) + \frac{C'\|\theta\|_2^2}{\sqrt{n}}\bigg) \ge 1-2\exp(-10r_{n,\lambda}^2)
    \end{align*}
\end{lem}
\begin{proof}
    Recall that $\delta_j \stackrel{d}{=}Q_{2\lambda}(\theta+\varepsilon_j+\tau_j)-\theta$ for $j\in [n]$, and notice that the entries of $\delta_2,\cdots,\delta_n$ are independent, and  entries of $\frac{1}{n-1}\sum_{j=2}^n\delta_j$ have $O(\frac{\lambda}{\sqrt{n}})$ sub-Gaussian norms:  
    \begin{align*}
        \max_{i\in[p]}\Big\|\frac{1}{n-1}\sum_{j=2}^n \delta_{ji}\Big\|_{\psi_2} \lesssim \frac{\lambda}{\sqrt{n}}. 
    \end{align*}
    Therefore, by Lemma \ref{psi1psi2} we have  
    \begin{align*}
        \max_{i\in [p]}\bigg\|\,\Big(\frac{1}{n-1}\sum_{j=2}^n\delta_{ji}\Big)^2\,\bigg\|_{\psi_1} \lesssim \frac{\lambda^2}{n},
    \end{align*}
    and thus Bernstein's inequality (Lemma \ref{bern}) gives, for any $t\ge 0$, that
    \begin{align}\label{bernsteinss}
        \mathbb{P}\bigg(\bigg\|\frac{1}{n-1}\sum_{j=2}^n\delta_j\bigg\|_2^2  \ge \mathbb{E}\bigg\|\frac{1}{n-1}\sum_{j=2}^n\delta_j\bigg\|_2^2+ t\bigg) \le 2\exp\left(-c \min\left\{\frac{n^2t^2}{p\lambda^4},\frac{nt}{\lambda^2}\right\}\right).
    \end{align}
    We set $t = C\max\{\frac{\lambda^2\sqrt{p}\cdot r_{n,\lambda}}{n},\frac{\lambda^2r_{n,\lambda}^2}{n}\}$ with large enough $C$ to establish that 
    \begin{align} 
        &\mathbb{P}\bigg(\bigg\|\frac{1}{n-1}\sum_{j=2}^n\delta_j\bigg\|_2^2  \ge \mathbb{E}\bigg\|\frac{1}{n-1}\sum_{j=2}^n\delta_j\bigg\|_2^2+ C\max\Big\{\frac{\lambda^2\sqrt{p}\cdot r_{n,\lambda}}{n},\frac{\lambda^2r_{n,\lambda}^2}{n}\Big\}\bigg)\nn\\ &\le 2\exp(-10r_{n,\lambda}^2). \label{sumnormbb}
    \end{align}
    It remains to compute the expectation. In light of the i.i.d. zero-mean $\delta_j$'s,  
    \begin{align}\label{expnorm}
        \mathbb{E}\bigg\|\frac{1}{n-1}\sum_{j=2}^n\delta_j\bigg\|_2^2 = \sum_{i=1}^p \mathbb{E}\bigg|\frac{1}{n-1}\sum_{j=2}^n\delta_{ji}\bigg|^2 = \frac{1}{n-1}\sum_{i=1}^p \mathbb{E}[\delta_{ji}^2]. 
    \end{align}
    Furthermore, we have 
    \begin{align}\nn
        \mathbb{E}[\delta_{ji}^2] &= \mathbb{E}\Big|Q_{2\lambda}(\theta_i+\varepsilon_{ji}+\tau_{ji})-\theta _i\Big|^2 = \mathbb{E}|Q_{2\lambda}(\theta_i+\varepsilon_{ji}+\tau_{ji})|^2 -\theta_i^2\\\explain \textrm{by $\mathbb{E}[Q_{2\lambda}(\theta_i+\varepsilon_{ji}+\tau_{ji})]=\theta_i$ (see Lemma \ref{lem:quannoi})} \\\label{hihihi}
        & \le\lambda^2 + \underbrace{\mathbb{E}\Big(|Q_{2\lambda}(\theta_i+\varepsilon_{ji}+\tau_{ji})|^2 \mathbf{1}(|\varepsilon_{ji}|\ge \sigma\sqrt{2(1+\nu)\log(np)})\Big)}_{:=T_1}-\theta_j^2.\\
        \explain \textrm{by $\mathbb{E}\Big(|Q_{2\lambda}(\theta_i+\varepsilon_{ji}+\tau_{ji})|^2 \mathbf{1}(|\varepsilon_{ji}|<\sigma\sqrt{2(1+\nu)\log(np)})\Big)\le \lambda^2$}\\
        &\quad ~\textrm{which holds because under (\ref{lambdacon1}), $|\theta_i+\varepsilon_{ji}+\tau_{ji}|\le 2\lambda$, and $Q_{2\lambda}(a)=\pm\lambda$ when $|a|\le 2\lambda$}\nn
    \end{align}
    Moreover,    
    \begin{align}\nn
        T_1&\le \mathbb{E}\Big((3\lambda+ |\varepsilon_{ji}|)^2 \mathbf{1}(|\varepsilon_{ji}|\ge \sigma\sqrt{2(1+\nu)\log(np)})\Big) \\\explain \textrm{by $|Q_{2\lambda}(\theta_i+\varepsilon_{ji}+\tau_{ji})|\le |\theta_i+\varepsilon_{ji}+2\lambda|\le 3\lambda+|\varepsilon_{ji}|$} 
        \\&\le 18\lambda^2 \mathbb{P}\big(|\varepsilon_{ji}|\ge \sigma\sqrt{2(1+\nu)\log(np)}\big) + 2\mathbb{E}\big[\varepsilon_{ji}^2\mathbf{1}(|\varepsilon_{ji}|\ge \sigma\sqrt{2(1+\nu)\log(np)})\big]\nn\\\explain\textrm{by $(A_1+A_2)^2\le 2A_1^2+2A_2^2$}
        \\& \le \frac{C'\lambda^2}{(np)^{1+\nu}} .\label{T1bound111}\\
        \explain\textrm{by   Lemma \ref{lem:sgtail} and $\mathbb{E}\big[\varepsilon_{ji}^2\mathbf{1}(|\varepsilon_{ji}|\ge \sigma\sqrt{2(1+\nu)\log(np)})\big]\lesssim\frac{\lambda^2}{(np)^{1+\nu}}$ (see the following)}
    \end{align}
    We now explain $\mathbb{E}\big[\varepsilon_{ji}^2\mathbf{1}(|\varepsilon_{ji}|\ge \sigma\sqrt{2(1+\nu)\log(np)})\big]\lesssim \frac{\lambda^2}{(np)^{1+\nu}}$: 
        \begin{align*}
        &\mathbb{E}\left[\varepsilon_{ji}^2\mathbf{1}(|\varepsilon_{ji}|\ge \sigma\sqrt{2(1+\nu)\log(np)})\right]\\
        &=\int_0^\infty \mathbb{P}\left(\varepsilon_{ji}^2\mathbf{1}(|\varepsilon_{ji}|\ge \sigma\sqrt{2(1+\nu)\log(np)})\ge t\right)dt\\\explain\textrm{by tail integral representation of expectation}\\
    &=\int^\infty_{2\sigma^2(1+\nu)\log(np)}\mathbb{P}\left(| \varepsilon_{ji}|^2\ge t\right)dt\\&\le \int^\infty_{2\sigma^2(1+\nu)\log(np)} 2\exp\Big(-\frac{t}{2\sigma^2}\Big)\,dt\\\explain\textrm{by Lemma \ref{lem:sgtail}}\\
        &  = \frac{4\sigma^2}{(np)^{1+\nu}}\le \frac{4\lambda^2}{(np)^{1+\nu}}. 
    \end{align*}
Substituting (\ref{T1bound111}) into (\ref{hihihi}) yields 
\[\mathbb{E}[\delta_{ji}^2]\le \big(1+\frac{C'}{np}\big)\lambda^2-\theta_j^2,\quad \forall i\in[p]\] and substituting this into (\ref{expnorm}) yields 
    \begin{align}\label{hithis2}
           \mathbb{E}\bigg\|\frac{1}{n-1}\sum_{j=2}^n\delta_j\bigg\|_2^2 \le \frac{1}{n-1}\left[p\lambda^2\Big(1+\frac{C'}{np}\Big)-\|\theta\|_2^2\right].
    \end{align}
In light of $r_{n,\lambda}=\frac{\|\theta\|_2^2/\lambda^2}{\sqrt{\|\theta\|_2^2/\lambda^2+p/n}}\asymp \min\{\frac{\|\theta\|_2}{\lambda},\frac{\|\theta\|_2^2}{\lambda^2}\sqrt{\frac{n}{p}}\}$ and therefore 
    \begin{gather}\label{hithis1}
        \frac{\lambda^2\sqrt{p}\cdot r_{n,\lambda}}{n} =O\Big(\frac{\|\theta\|_2^2}{\sqrt{n}}\Big)\qquad\text{and}\qquad \frac{\lambda^2 r_{n,\lambda}^2}{n} =O\Big(\frac{\|\theta\|_2^2}{n}\Big). 
    \end{gather}
    Substituting (\ref{hithis2}) and (\ref{hithis1}) into (\ref{sumnormbb}) completes the proof. 
\end{proof}
\begin{lem}
    \label{lem:Pstarstar}
   In the proof of Theorem \ref{maintheorem} in Appendix \ref{app:proofthm1}, we have
    \begin{align*}
        P^{**}\le  \mathbb{P}\bigg(\langle\tilde{\delta}_1,\hat{\xi}\rangle < - \frac{(1-\epsilon_{n,\lambda})\|\theta\|_2^2}{\sqrt{\|\theta\|_2^2+\frac{p\lambda^2}{n}}}\bigg) + \frac{1}{n^{1+\nu}p^\nu}
    \end{align*}
    for some $\epsilon_{n,\lambda}=O(r_{n,\lambda}^{-1}+(\frac{\log n}{n})^{1/4})$.
\end{lem}
\begin{proof}
    Recall that $\delta_1=Q_{2\lambda}(\theta+\varepsilon_1+\tau_1)-\theta$ and $\tilde{\delta}_1 = \lambda\sign(\theta+\varepsilon_1+\tau_1)-\mathbb{E}[\lambda\sign(\theta+\varepsilon_1+\tau_1)]$. We define the  event 
    \begin{align}\label{evehatB2}
        \hat{B}_2:=\left\{\max_{i\in[p]}|\theta_i+\varepsilon_{1i}|\le \lambda\right\}.
    \end{align}
     Under $\lambda\ge \|\theta\|_\infty + \sigma\sqrt{2(1+\nu)\log(np)}$, Lemma \ref{lem:sgtail} gives 
    \begin{align}\label{hatB2prob}
        \mathbb{P}(\hat{B}_2)\ge \mathbb{P}\left(\max_{i\in[p]}|\varepsilon_{1i}|\le \sigma\sqrt{2(1+\nu)\log(np)}\right) \ge 1- \frac{p}{(np)^{1+\nu}}=1-\frac{1}{n^{1+\nu}p^\nu}.
    \end{align}
    Moreover, for any $i\in[p]$,  
    \begin{align*}
        &|\mathbb{E}[\lambda\sign(\theta_i+\varepsilon_{1i}+\tau_{1i})]-\theta_i| \\
        &= |\mathbb{E}[\lambda\sign(\theta_i+\varepsilon_{1i}+\tau_{1i})-(\theta_i+\varepsilon_{1i})]|\\
        &= |\mathbb{E}[\lambda\sign(\theta_i+\varepsilon_{1i}+\tau_{1i})-(\theta_i+\varepsilon_{1i})]\mathbf{1}(|\theta_i+\varepsilon_{1i}|\ge  \lambda)|\\ \explain\textrm{by $\mathbb{E}_{\tau_{1i}}[\lambda\sign(\theta_i+\varepsilon_{1i}+\tau_{1i})\mathbf{1}(|\theta_i+\varepsilon_{1i}|<  \lambda)$} \\&\quad\quad~\textrm{$=\mathbb{E}_{\tau_{1i}}[Q_{2\lambda}(\theta_i+\varepsilon_{1i}+\tau_{1i})\mathbf{1}(|\theta_i+\varepsilon_{1i}|<  \lambda)]$}\\&\qquad~\textrm{$=(\theta_i+\varepsilon_{1i})\mathbf{1}(|\theta_i+\varepsilon_{1i}|<\lambda)$; see Lemma \ref{lem:quannoi}}\\
        &\le \mathbb{E}\Big[(\lambda+|\theta_i|+|\varepsilon_{1i}|)\mathbf{1}\Big(|\varepsilon_{1i}|\ge \sigma\sqrt{2\log(np)}\Big)\Big]\\
        \explain\textrm{by triangle inequality and (\ref{lambdacon1})}\\
        &\le \frac{2\lambda}{np} + \int_{\sigma\sqrt{2\log(np)}}^\infty \mathbb{P}(|\varepsilon_{1i}|\ge t)\,dt\\
        &\le \frac{2\lambda}{np}+\int_{\sigma\sqrt{2\log(np)}}^\infty 2\exp\Big(-\frac{t^2}{2\sigma^2}\Big)\,dt\le \frac{3\lambda}{np}.\\
        \explain\textrm{by Lemma \ref{lem:sgtail}}
    \end{align*}
    Combining with $\|\theta\|_2\ge \lambda$ leads to  
    \begin{align}\label{highernormbound}
        \|\mathbb{E}[\lambda \sign(\theta+\varepsilon_1+\tau_1)]-\theta\|_2\le \frac{3\lambda}{n\sqrt{p}}\le \frac{6}{n} \frac{\|\theta\|_2^2}{\sqrt{\|\theta\|_2^2+\frac{p\lambda^2}{n}}}
    \end{align}
    In view of Equation (\ref{Pstarstardefine}), for some positive $\epsilon_{n,\lambda},\epsilon_{n,\lambda}'=O(r_{n,\lambda}^{-1}+(\frac{\log n}{n})^{1/4})$, we have that 
    \begin{align*}
       & P^{**} = \mathbb{P}\left(\langle Q_{2\lambda}(\theta+\varepsilon_1+\tau_1)-\theta,\hat{\xi}\rangle <-(1-\epsilon_{n,\lambda})\frac{\|\theta\|_2^2}{\sqrt{\|\theta\|_2^2+\frac{p\lambda^2}{n}}}\right)\\
        &\le \mathbb{P}\left(\langle Q_{2\lambda}(\theta+\varepsilon_1+\tau_1)-\theta,\hat{\xi}\rangle <-(1-\epsilon_{n,\lambda})\frac{\|\theta\|_2^2}{\sqrt{\|\theta\|_2^2+\frac{p\lambda^2}{n}}},~\hat{B}_2\right) + \frac{1}{n^{1+\nu}p^\nu}\\\explain\textrm{by Equation (\ref{hatB2prob})}\\
        &\le \mathbb{P}\bigg(\langle \lambda\sign(\theta+\varepsilon_1+\tau_1)-\mathbb{E}[\lambda\sign(\theta+\varepsilon_1+\tau_1)],\hat{\xi}\rangle \\&\qquad\quad<-(1-\epsilon_{n,\lambda})\frac{\|\theta\|_2^2}{\sqrt{\|\theta\|_2^2+\frac{p\lambda^2}{n}}}+\|\mathbb{E}[\lambda\sign(\theta+\varepsilon_1+\tau_1)]-\theta\|_2\bigg)+ \frac{1}{n^{1+\nu}p^\nu}\\\explain\textrm{by \textrm{$Q_{2\lambda}(\theta+\varepsilon_1+\tau_1)=\lambda\sign(\theta+\varepsilon_1+\tau_1)$ holds on the event $\hat{B}_2$ in (\ref{evehatB2})}}\\
        &\le \mathbb{P}\left(\langle\delta_1,\hat{\xi}\rangle<-(1-\epsilon'_{n,\lambda})\frac{\|\theta\|_2^2}{\sqrt{\|\theta\|_2^2+\frac{p\lambda^2}{n}}}\right) + \frac{1}{n^{1+\nu}p^\nu}\,.\\\explain\textrm{by Equation (\ref{highernormbound})}
    \end{align*}
    This completes the proof. 
\end{proof}
\begin{lem}
    \label{Rthetamu}
    Let $\theta\sim {\rm Unif}(\Delta S^{p-1})$, then we have
    \begin{align}
        \mathbb{P}\left(\mu(\theta)\ge \sqrt{\frac{3\log p}{p}}\right)\le 1-2p^{-1/4}-2\exp(-cp).
    \end{align}
\end{lem}
\begin{proof}
    Since $\mu(\cdot)$ is invariant under rescaling, we can treat $\theta\sim {\rm Unif}(S^{p-1})$ only. It is well known that $\theta \stackrel{{\rm d}}{=}\frac{g}{\|g\|_2}$ for $g\sim N(0,I_p)$. By standard Gaussian tail bound, along with a union bound, it is not hard to show $\|g\|_\infty\le \sqrt{2.5\log p}$ with probability at least $1-2p^{-1/4}$. Furthermore, it is standard to show that $\|g\|_2\ge 1.01\sqrt{p}$ with probability at least $1-2\exp(-cp)$ for some absolute constant $c$ (e.g., \cite[Theorem 3.1.1]{vershynin2018high}). Combining both bounds completes the proof. 
\end{proof}
 \begin{lem}[Generalization of Lemma \ref{utubound} for Appendix \ref{app:proofthm2}] \label{lem:genutu}Given $R\in \mathbb{O}(p)$. Suppose that $U_1\in R^{p\times n}$ has independent zero-mean entries of   sub-Gaussian norms bounded by $A_1$; when conditioning on $U_1$, $U_2\in R^{p\times n}$ has independent zero-mean rows of sub-Gaussian norms bounded by $A_2$. Then $U=RU_1+U_2$ satisfies 
 \begin{align*}
     \mathbb{P}\left(\|U^TU-\mathbb{E}(U^TU)\|_{op}\le  C(A^2_1+A_2^2)\cdot n\max\{1,\sqrt{\frac{p}{n}}\}\right)\le 8\exp(-2n). 
 \end{align*} \end{lem}
 \begin{proof}
     By $U^TU = U_1^TU_1+U_2^TU_2+U_2^TRU_1+U_1^TR^TU_2$ and $\mathbb{E}(U_2^TRU_1)=0$ (conditioning on $U_1$ and using randomness of $U_2$), triangle inequality yields
     \begin{align*}
         \|U^TU-\mathbb{E}(U^TU)\|_{op} \le \|U_1^TU_1-\mathbb{E}(U_1^TU_1)\|_{op}+\|U_2^TU_2-\mathbb{E}(U_2^TU_2)\|_{op}+ 2\|U_2^TRU_1\|_{op}.
     \end{align*}
     By Lemma \ref{utubound}, with probability at least $1-4\exp(-2n)$, we have 
     \begin{align*}
         \|U_1^TU_1-\mathbb{E}(U_1^TU_1)\|_{op} \lesssim A_1^2n\max\{1,\sqrt{\frac{p}{n}}\}\,,\quad \|U_2^TU_2-\mathbb{E}(U_2^TU_2)\|_{op} \lesssim A_2^2n\max\{1,\sqrt{\frac{p}{n}}\}.
     \end{align*}
     It remains to bound $\|U_2^TRU_1\|_{op}$. By a standard bound on random matrix with independent, zero-mean sub-Gaussian entries (e.g., \cite[Theorem 4.4.5]{vershynin2018high}), we have that
     \begin{align*}
         \mathbb{P}\Big(\|U_1\|_{op}\le C'A_1(\sqrt{n}+\sqrt{p}) \Big) \ge 1-2\exp(-2n). 
     \end{align*}
     To bound $\|U_2^TRU_1\|_{op}$, we assume that $\|U_1\|_{op}\le C'A_1(\sqrt{n}+\sqrt{p})$ holds, and we utilize the randomness of $U_2$ by conditioning on $U_1$. Let $N_{1/4}$ be a $1/4$-net of $S^{n-1}$ of cardinality smaller than $9^n$, then   a simple discretization of the operator norm (e.g.,   \cite[Exercise 4.4.3]{vershynin2018high}) yields
     \begin{align}\label{discretebound}
         \|U_2^TRU_1\|_{op}\le 2\sup_{v_1,v_2\in N_{1/4}} (U_2v_2)^TRU_1v_1.
     \end{align}
     For any pair $(v_1,v_2)\in N_{1/4}\times N_{1/4}$, we have $\|RU_1v_1\|_2\le \|U_1\|_{op}\le C 'A_1(\sqrt{n}+\sqrt{p})$, and by the assumption on $U_2$ (conditioning on $U_1$) we also have that the entries of $U_2v_2$ are independent, zero-mean, and of sub-Gaussian norms bounded by $A_2$. Therefore, $\|U_2v_2\|_{\psi_2} =O(A_2)$ (see Lemma \ref{ver261}). Taken collectively, we reach $\|(U_2v_2)^TRU_1v_1\|_{\psi_2}=O(A_1A_2(\sqrt{n}+\sqrt{p}))$. Thus, the standard sub-Gaussian tail bound, along with a union bound over $(v_1,v_2)\in N_{1/4}\times N_{1/4}$, establishes
     \begin{align*}
         \mathbb{P}\left(\sup_{v_1,v_2\in N_{1/4}}(U_2v_2)^TRU_1v_1\ge C'' A_1A_2(\sqrt{n}+\sqrt{p})t\right) \le 2\exp\Big(2n\log 9-ct^2\Big),\quad \forall t\ge 0.
     \end{align*}
     Setting $t\asymp \sqrt{n}$ and recalling (\ref{discretebound}), we obtain that $\|U_2^TRU_1\|_{op}\lesssim C''A_1A_2 n\max\{1,\sqrt{\frac{p}{n}}\}$ holds with probability at least $1-2\exp(-2n)$. Putting all the pieces together concludes the proof. 
\end{proof}
\begin{lem}[Adaptation of Lemma \ref{lem:boundsumvnorm} for Appendix \ref{app:proofthm2}] \label{lem:sumnormlem}
    In the proof of Theorem \ref{thm:spifree} in Appendix \ref{app:proofthm2} with some given $R\in \mathbb{O}(p)$, where $\delta_j=Q_{2\lambda}(\theta_R+R\varepsilon_j+\tau_j)-\theta_R$ for $j\in[n]$, we have
    \begin{align*}
        \mathbb{P}\bigg(\bigg\|\frac{1}{n-1}\sum_{j=2}^n\delta_j\bigg\|_2^2<\frac{p\lambda^2}{n}\Big(1+\frac{C}{n}\Big) + \frac{C'\|\theta\|_2^2}{\sqrt{n}}\bigg)\ge 1- 6\exp(-10r_{n,\lambda}^2).
    \end{align*}
\end{lem}
\begin{proof} 
    We let $\delta_j':=Q_{2\lambda}(\theta_R+R\varepsilon_j+\tau_j)-(\theta_R+R\varepsilon_j)$ and thus $\delta_j=\delta_j'+R\varepsilon_j$, yielding 
    \begin{align}\label{inisumnorm11}
        \bigg\|\frac{1}{n-1}\sum_{j=2}^n\delta_j\bigg\|_2^2 =\bigg\|\frac{1}{n-1}\sum_{j=2}^n\delta_j'\bigg\|_2^2 + \bigg\|\frac{1}{n-1}\sum_{j=2}^n\varepsilon_j\bigg\|_2^2 + 2\bigg\langle \frac{1}{n-1}\sum_{j=2}^n\delta_j', \frac{1}{n-1}\sum_{j=2}^nR\varepsilon_j\bigg\rangle.
    \end{align}
    Note that $\frac{1}{n-1}\sum_{j=2}^n\delta_j'$ and $\frac{1}{n-1}\sum_{j=2}^n\varepsilon_j$ have independent, zero-mean entries of $O(\frac{\lambda}{\sqrt{n}})$ sand $O(\frac{\sigma}{\sqrt{n}})$ sub-Gaussian norms, respectively. By Bernstein's inequality as in Equations (\ref{bernsteinss})--(\ref{sumnormbb}),   with probability at least $1-4\exp(-10r_{n,\lambda}^2)$ we have 
    \begin{gather}\label{sumbound11}
        \bigg\|\frac{1}{n-1}\sum_{j=2}^n\delta_j'\bigg\|_2^2\le \mathbb{E}\bigg\|\frac{1}{n-1}\sum_{j=2}^n\delta_j'\bigg\|_2^2 + C\max\bigg\{\frac{\lambda^2\sqrt{p}\cdot r_{n,\lambda}}{n},\frac{\lambda^2r_{n,\lambda}^2}{n}\bigg\},
        \\ \label{sumbound22}\bigg\|\frac{1}{n-1}\sum_{j=2}^n\varepsilon_j\bigg\|_2^2\le \mathbb{E}\bigg\|\frac{1}{n-1}\sum_{j=2}^n\varepsilon_j\bigg\|_2^2 + C\max\bigg\{\frac{\sigma^2\sqrt{p}\cdot r_{n,\lambda}}{n},\frac{\sigma^2r_{n,\lambda}^2}{n}\bigg\}.
    \end{gather}
    We assume  these hold. Combining with $\mathbb{E}\|\frac{1}{n-1}\sum_{j=2}^n\varepsilon_j\|_2^2=O(\frac{\sigma^2p}{n})$, we have
    \begin{align*}
         \bigg\|\frac{1}{n-1}\sum_{j=2}^n\varepsilon_j\bigg\|_2^2 \lesssim \frac{\sigma^2(p+r_{n,\lambda}^2)}{n}.
    \end{align*}  
    Recalling that $\frac{1}{n-1}\sum_{j=2}^n\delta_j'$ has independent zero-mean entries of $O(\frac{\lambda}{\sqrt{n}})$ sub-Gaussian norms, we have
    \begin{align*}
        \bigg\|\bigg\langle \frac{1}{n-1}\sum_{j=2}^n\delta_j', \frac{1}{n-1}\sum_{j=2}^nR\varepsilon_j\bigg\rangle\bigg\|_{\psi_2} \lesssim \frac{\lambda}{\sqrt{n}}\bigg\|\frac{1}{n-1}\sum_{j=2}^n\varepsilon_j\bigg\|_2 \lesssim \frac{\sigma\lambda (\sqrt{p}+r_{n,\lambda})}{n},
    \end{align*}
    which implies that
    \begin{align}\label{crossbound}
        \bigg|\bigg\langle \frac{1}{n-1}\sum_{j=2}^n\delta_j', \frac{1}{n-1}\sum_{j=2}^nR\varepsilon_j\bigg\rangle\bigg| \lesssim\frac{\sigma\lambda(\sqrt{p}r_{n,\lambda}+r_{n,\lambda}^2)}{n} 
    \end{align}
    holds with probability at least $1-2\exp(-10r_{n,\lambda}^2)$.

    We now bound \begin{align}
        \mathbb{E}\bigg\|\frac{1}{n-1}\sum_{j=2}^n\delta_j'\bigg\|_2^2=\sum_{i=1}^p\mathbb{E}\bigg|\frac{1}{n-1}\sum_{j=2}^n \delta'_{ji}\bigg|^2=\sum_{i=1}^p\sum_{j=2}^n\frac{\mathbb{E}[\delta'_{ji}]^2}{(n-1)^2}=\sum_{i=1}^p\frac{\mathbb{E}[\delta'_{ji}]^2}{n-1},\label{exptobound}
    \end{align}
    where the last two inequalities hold because $\delta'_j$s are independent and zero-mean.
    Letting $R_i^T$ be the $i$-th row of $R$, we expect over $\tau_{ji}$ and $\varepsilon_j$ in sequence to reach
    \begin{align}\nn
        &\mathbb{E}[\delta'_{ji}]^2 = \mathbb{E}\Big[Q_{2\lambda}(\theta_{R,i}+R_i^T\varepsilon_j+\tau_{ji})-(\theta_{R,i}+R_i^T\varepsilon_j)\Big]^2\\
        \nn& = \mathbb{E}_{\varepsilon_j}\mathbb{E}_{\tau_{ji}}\left[Q_{2\lambda}(\theta_{R,i}+R_i^T\varepsilon_j+\tau_{ji})^2+(\theta_{R,i}+R_i^T\varepsilon_j)^2-2(\theta_{R,i}+R_i^T\varepsilon_j)Q_{2\lambda}(\theta_{R,i}+R_i^T\varepsilon_j+\tau_{ji})\right] \\
        \nn& = \mathbb{E}_{\varepsilon_j}\left[\mathbb{E}_{\tau_{ji}}\big[Q_{2\lambda}(\theta_{R,i}+R_i^T\varepsilon_j+\tau_{ji})^2\big]-(\theta_{R,i}+R_i^T\varepsilon_j)^2\right]\\\explain\textrm{by Lemma \ref{lem:quannoi}}\\
        & = \mathbb{E}\big[Q_{2\lambda}(\theta_{R,i}+R_i^T\varepsilon_j+\tau_{ji})^2\big]-\theta_{R,i}^2 - \mathbb{E}_{\varepsilon_j}(R_i^T\varepsilon_j)^2  \label{singleentry}
    \end{align}
    Moreover, by $|Q_{2\lambda}(a)-a|\le \lambda$ and  , we have
    \begin{align*}
        &\mathbb{E}[Q_{2\lambda}(\theta_{R,i}+R_i^T\varepsilon_j+\tau_{ji})^2]  \le \lambda^2 + \mathbb{E}\big[Q_{2\lambda}(\theta_{R,i}+R_i^T\varepsilon_j+\tau_{ji})^2\mathbf{1}(|\theta_{R,i}+R_i^T\varepsilon_j|>\lambda)\big]\\\explain\textrm{by $|Q_{2\lambda}(a+\tau_{ji})\mathbf{1}(|a|\le \lambda)|=\lambda$ for any $a\in \mathbb{R}$}\\
        & \le \lambda^2  + \mathbb{E}\left[\Big(2\lambda+|\theta_{R,i}|+|R_i^T\varepsilon_j|\Big)^2\cdot\mathbf{1}\Big(|R_i^T\varepsilon_j|\ge \sigma\sqrt{2(1+\nu)\log(np)}\Big)\right]\\\explain\textrm{by Equation (\ref{lambdahaar})}\\
        & \le \lambda^2 + 18\lambda^2 \mathbb{P}\Big(|R_i^T\varepsilon_j|\ge \sigma\sqrt{2(1+\nu)\log(np)}\Big) \\&\qquad\qquad + 2 \mathbb{E}\left[|R_i^T\varepsilon_j|^2\mathbf{1}\Big(|R_i^T\varepsilon_j|\ge \sigma\sqrt{2(1+\nu)\log(np)}\Big)\right]\\\explain\textrm{by $|\theta_{R,i}|\le\lambda$ and $(A+B)^2\le 2A^2+2B^2$}\\
        & \le \left(1+\frac{18}{(np)^{1+\nu}}\right)\lambda^2 + 2 \mathbb{E}\left[|R_i^T\varepsilon_j|^2\mathbf{1}\Big(|R_i^T\varepsilon_j|\ge \sigma\sqrt{2(1+\nu)\log(np)}\Big)\right].\\\explain\textrm{by Lemma \ref{lem:sgtail}}
    \end{align*}
    By the tail bound of $|R_i^T\varepsilon_j|$ (cf. Lemma \ref{lem:sgtail}), we have
    \begin{align*}
        &\mathbb{E}\left[|R_i^T\varepsilon_j|^2\mathbf{1}\Big(|R_i^T\varepsilon_j|\ge \sigma\sqrt{2(1+\nu)\log(np)}\Big)\right]\\
        &=\int_0^\infty \mathbb{P}\left(|R_i^T\varepsilon_j|^2\mathbf{1}\Big(|R_i^T\varepsilon_j|\ge \sigma\sqrt{2(1+\nu)\log(np)}\Big)\ge t\right)dt\\\explain\textrm{by tail integral representation of expectation}\\
    &=\int^\infty_{2\sigma^2(1+\nu)\log(np)}\mathbb{P}\left(|R_i^T\varepsilon_j|^2\ge t\right)dt\\
        & \le 2\int^\infty_{2\sigma^2(1+\nu)\log(np)}\exp\Big(-\frac{t}{2\sigma^2}\Big)\,dt = \frac{4\sigma^2}{(np)^{1+\nu}}.\\
        \explain\textrm{by Lemma \ref{lem:sgtail}}
    \end{align*}
    Combining the preceding two displays establishes $
        \mathbb{E}\,[Q_{2\lambda}(\theta_{R,i}+R_i^T\varepsilon_j+\tau_{ji})^2] \le \left(1+\frac{20}{(np)^{1+\nu}}\right) \lambda^2$, which together with (\ref{singleentry}) and (\ref{exptobound}) yields
        \begin{align}
              \mathbb{E}\bigg\|\frac{1}{n-1}\sum_{j=2}^n\delta_j'\bigg\|_2^2 \le \frac{1}{n-1}\left[(1+\frac{1}{n})p\lambda^2-\|\theta\|_2^2-\mathbb{E}\|\varepsilon_j\|_2^2\right]. \label{expnormrr}
        \end{align}
        We now substitute Equations (\ref{expnormrr}), (\ref{crossbound}), (\ref{sumbound11})--(\ref{sumbound22}) into Equation (\ref{inisumnorm11}), along with noticing $\mathbb{E}\|\frac{1}{n-1}\sum_{j=2}^n \varepsilon_j\|_2^2 = \frac{1}{n-1}\mathbb{E}\|\varepsilon_j\|_2^2$ in Equation (\ref{sumbound22}) and $r_{n,\lambda}\lesssim \min\{\frac{\|\theta\|_2}{\lambda},\frac{\|\theta\|_2^2}{\lambda^2}\sqrt{\frac{n}{p}}\}$, to 
        establish
        \begin{align}\nn
               \bigg\|\frac{1}{n-1}\sum_{j=2}^n\delta_j\bigg\|_2^2 \le \frac{p\lambda^2}{n}\Big(1+\frac{C}{n}\Big) + \frac{C'\|\theta\|_2^2}{\sqrt{n}}
        \end{align} 
        with the promised probability. 
\end{proof}
\begin{lem}
    [Adaptation of Lemma \ref{lem:Pstarstar} for    Appendix \ref{app:proofthm2}] \label{lem:deltatotilehaar}In the proof of Theorem \ref{thm:spifree} in   Appendix \ref{app:proofthm2}, we have
    \begin{align*}
        P^{**}\le \mathbb{P}\left(\frac{\langle\tilde{\delta}_1',\hat{\xi}\rangle}{\lambda}\le -(1-\epsilon_{n,\lambda})r_{n,\lambda}\right) + \frac{1}{n^{1+\nu}p^\nu}
    \end{align*}
    for some $\epsilon_{n,\lambda}=O(r_{n,\lambda}^{-1}+(\frac{\log n}{n})^{1/4})+\frac{1}{\log^{1/4}(np)}$. 
\end{lem}
\begin{proof}
    Recall that $P^{**}$ is defined in (\ref{Psssdefine}) with $\delta_1':=Q_{2\lambda}(\theta_R+R\varepsilon_1+\tau_1)-(\theta_R+R\varepsilon_1)$, and as a surrogate we introduce $\tilde{\delta}_1'=\lambda\sign(\theta_R+R\varepsilon_1+\tau_1)-\mathbb{E}_{\tau_1}[\lambda\sign(\theta_R+R\varepsilon_1+\tau_1)]$. Notice that on the event $\hat{B}_2=\{\|\theta_R+R\varepsilon_1\|_\infty\le \lambda\}$, we have $\delta_1'=\tilde{\delta}_1'$. Also, by $\lambda\ge (\frac{3\log p}{p})^{1/2}+ \sigma\sqrt{2(1+\nu)\log(np)}\ge \|\theta_R\|_\infty + \sigma\sqrt{2(1+\nu)\log(np)}$, Lemma \ref{lem:sgtail} along with a union bound gives
    \begin{align*}
        \mathbb{P}(\hat{B}_2) \ge \mathbb{P}\Big(\|R\varepsilon_1\|_\infty \le \sigma\sqrt{2(1+\nu)\log(np)}\Big)\ge 1-\frac{1}{n^{1+\nu}p^\nu}.
    \end{align*}
    Therefore, for some $\epsilon_{n,\lambda}=O(r_{n,\lambda}^{-1}+(\frac{\log n}{n})^{1/4})+\frac{1}{\log^{1/4}(np)}$, we have 
    \begin{align}\nn
        &P^{**}= \mathbb{P}\left(\frac{\langle\tilde{\delta}_1',\hat{\xi}\rangle}{\lambda}\le-(1-\epsilon_{n,\lambda})r_{n,\lambda}\right) \le \mathbb{P}\left(\frac{\langle \delta_1',\hat{\xi}\rangle}{\lambda}\le-(1-\epsilon_{n,\lambda}r_{n,\lambda}),\hat{B}_2\right)+\mathbb{P}(\hat{B}_2^c),
    \end{align}
    which then yields the claim immediately. 
\end{proof}

\section{Supplemental for Section \ref{sec:experi}}\label{app:detailexp}

\subsection{Further Details (Figures \ref{fig:sub1}, \ref{fig:sub2})} \label{app:details}
The main aim of Figures \ref{fig:sub1}, \ref{fig:sub2} is to corroborate Corollary \ref{cor:exact}. Note that separation condition (\ref{sepaexact}),   when ignoring the factor of $1+\epsilon$, reads
\begin{align}
    \Delta^2 \ge \lambda^2 \bigg(1+\sqrt{1+\frac{2p}{n\log n}}\bigg)\log n.\label{separepr}
\end{align}
Also, the minimum of $\lambda$ in Equation (\ref{lambdacon1}) is approximately $\sigma\sqrt{2\log (np)}=\sqrt{2\log (np)}$, by ignoring the factor of $1+\nu$ and observing that $\|\theta\|_\infty$ is typically dominated by $\sqrt{\log(np)}$. We therefore  set $\lambda$ approximately as $\sqrt{2\log(np)}$. 

\subsubsection{Low-Dimensional Setting (Figure \ref{fig:sub1})}
Recall that we set $a=\log n$, $b=\Delta$, and fix \[p=5.\] In our experiment, given $a$ we set $n=\lceil \exp(a)\rceil$.

Since we test $a\in[3.4,6.9]$, the range of $n$ is roughly $[30,1000]$. Hence, under $p=5$, we make further simplifications:
\begin{itemize}[leftmargin=5ex,topsep=0.25ex]
    \setlength\itemsep{-0.3em}
    \item Since $\sqrt{2\log(np)}\approx\sqrt{2\log n}$,   in this experiment we simply set $\lambda=\sqrt{2\log n}$;
    \item Since $\frac{p}{n\log n}\ll 1$, approximately, the separation condition (\ref{separepr}) further reduces to 
    \[\Delta^2\ge 2\lambda^2\log n = (2\log n)^2,\]
    that is,
    \[\Delta=2\log n.\]
    As such, the separation condition is roughly $b=2a.$ 
\end{itemize}

\subsubsection{High-Dimensional Setting (Figure \ref{fig:sub2})}
Recall that we set $n=100$, $p=bn\log n$, $\Delta^2=\lambda^2(1+\sqrt{a})\log n$, and $\lambda=\sqrt{2\log(np)}$. For this experiment, the setting provided in the main text is already clear. Hence, we only derive the separation condition (on $(a,b)$) here. In fact, substituting $\Delta^2=\lambda^2(1+\sqrt{a})\log n$ and $p=bn\log n$ into (\ref{separepr}) gives
\[\lambda^2(1+\sqrt{a})\ge \lambda^2\bigg(1+\sqrt{1+2b}\bigg)\log n,\]
which then reduces to 
\[a=1+2b.\]
We note that \cite{ndaoud2022sharp} adopted the same experimental design. 
\subsection{Shrinking the Dithering Level (Figure \ref{fig:sub3})}
\label{app:shrink}
Although (\ref{lambdacon1}) is required in our theoretical results, in practice we suggest $\lambda =s\sqrt{2\log (np)} $ for some shrinkage parameter $0<s<1$ to draw a better bias-and-variance tradeoff (cf. Remark \ref{rem:bias-variance}), in a spirit similar to \cite[Remark 1]{chen2025parameter} for two-bit covariance estimation.  Indeed, in light of  
\[r_{n,\lambda} = \frac{\|\theta\|_2^2 /\lambda^2}{\sqrt{\frac{\|\theta\|_2^2}{\lambda^2}+\frac{p}{n}}},\]
an   larger $\lambda$ leads to smaller $c_{n,\lambda}$. Hence, using overly large $\lambda$ leads to performance degradation. While (\ref{lambdacon1}) in Theorem \ref{maintheorem} reads  
\[\lambda\ge \|\theta\|_\infty + \sigma\sqrt{2(1+\nu)\log(np)}\approx \sigma\sqrt{2\log(np)},\]
and ensures that  $\lambda\ge \max_{i\in[n]}\|X_i\|_\infty$ with high probability,  in practice, a smaller $\lambda$ could be sufficient for rendering small bias, and in this case using a smaller $\lambda$ typically leads to better estimation performance. 

 
\end{document}